\documentclass[11pt,a4paper]{article}

\usepackage{cite}
\usepackage{amsmath,amssymb}
\usepackage{amsthm}
\usepackage{microtype}
\usepackage{framed}
\usepackage[top=1in, bottom=1in, left=1in, right=1in]{geometry}

\usepackage{array}
\usepackage{graphicx}
\usepackage{float}
\usepackage{url}
\usepackage{hyperref}
\usepackage{multirow}
\usepackage{color,soul}
\usepackage{xspace}
\usepackage{tabularx}
\usepackage{subfig}

\usepackage{enumitem}
\usepackage{algorithm}
\usepackage{framed}
\usepackage[noend]{algpseudocode}

\usepackage{color}
\xspaceaddexceptions{]\}}
\definecolor{byzantine}{rgb}{0.74, 0.2, 0.64}
\definecolor{cocoabrown}{rgb}{0.82, 0.41, 0.12}
	\definecolor{darkcyan}{rgb}{0.0, 0.55, 0.55}
\makeatletter
\newcommand{\xMapsto}[2][]{\ext@arrow 0599{\Mapstofill@}{#1}{#2}}
\def\Mapstofill@{\arrowfill@{\Mapstochar\Relbar}\Relbar\Rightarrow}
\makeatother
\makeatletter\renewcommand{\ALG@name}{Protocol}\makeatother

\newcommand{\vir}[1]{``{#1}''}
\definecolor{orange}{rgb}{1,0.5,0}
\definecolor{wildwatermelon}{rgb}{0.99, 0.42, 0.52}

\newcommand{\colono}{{\em particle}\xspace}

\newcommand{\settler}{\colono}

\newcommand{\inited}{{\em starting}\xspace}
\newcommand{\faulty}{{\em faulty}\xspace}
\newcommand{\faulties}{{\em faults}\xspace}
\newcommand{\marker}{{\em marker}\xspace}
\newcommand{\markers}{{\em markers}\xspace}
\newcommand{\explorer}{{\em explorer}\xspace}
\newcommand{\explorers}{{\em explorers}\xspace}
\newcommand{\recruiter}{{\em follower}\xspace}

\newcommand{\leader}{{\em candidate}\xspace}
\newcommand{\leaderplu}{{\em candidates}\xspace}

\newcommand{\chosen}{{\em leader}\xspace}

\newcommand{\slave}{{\em slave}\xspace}

\newcommand{\probe}{{\em probe}\xspace}
\newcommand{\probes}{{\em probes}\xspace}

\newcommand{\collector}{{\em collector}\xspace}
\newcommand{\follower}{\recruiter}

\newcommand{\opposer}{{\em opposer}\xspace}

\newcommand{\exodusleader}{{\em nofaulty.leader}\xspace}

\newcommand{\algo}{{\sc LineRecovery}\xspace}

\newtheorem{observation}{Observation}

\newtheorem{lemma}{Lemma}

\newtheorem{theorem}{Theorem}

\newcommand{\subriteDescription}[1]{\noindent{\bf #1.}}

\begin{document}

\title{Line-Recovery by Programmable Particles}

\author{Giuseppe Antonio Di Luna\footnotemark[1], \and Paola Flocchini\footnotemark[1], \and Giuseppe Prencipe\footnotemark[2],\and Nicola Santoro\footnotemark[3],\and Giovanni Viglietta\footnotemark[1]}
\def\thefootnote{\fnsymbol{footnote}}
 \footnotetext[1]{\noindent
 School Electrical Engineering and Computer Science, University of Ottawa,
 Canada.}
  \footnotetext[2]{\noindent
  Dipartimento di Informatica, University of Pisa, Italy}
\footnotetext[3]{\noindent
 School of Computer Science, Carleton University, 
Canada.}

\maketitle

\begin{abstract}
Shape  formation has been recently studied in distributed systems of  programmable 
particles.  In this paper we  consider  the {\em shape recovery} problem of restoring the shape when $f$ 
of the $n$ particles have crashed. We focus on the basic {\em line} shape, used as a tool for the construction 
of more complex configurations.

We present a solution to the {\em line recovery} problem by the non-faulty
anonymous particles; the solution works
regardless of the initial distribution and number  $f<n-4$ of faults,   of the local orientations of the non-faulty entities,
and of the number of non-faulty  entities activated in each round (i.e., semi-synchronous adversarial scheduler).

 \end{abstract}

\section{Introduction}

The problems arising in distributed systems composed of autonomous mobile computational entities
has been  extensively studied, in particular the class of {\em pattern formation} problems requiring the
entities to move in  the space where they operate  until, in finite time,  they  form a given pattern
  (modulo translation, rotation, scaling, and reflection), and terminate
 (e.g.,~\cite{AgP06,CiFPS12,FlPS12,xFlPSW08,xFuYKY16,YUKY15}).
 
Very recently, other types of distributed computational universes have been started to be examined (e.g., \cite{EmLUW14}),
most significantly those arising in the large inter-disciplinary field of   studies on {\em  programmable matter} ~\cite{xToM91},
that is matter that has the ability to change its physical properties and appearence
 (e.g., shape,  color, etc.) based on  user input or autonomous sensing. 
  Programmable matter is typically viewed as a  very large number of   very small  (possibly nano-level) computational particles that  are 
  programmed to  collectively perform some global  task by means of local interactions. 
 Such particles  could have applications 
 in a variety of  important situations: smart materials, autonomous monitoring and repair, minimal invasive surgery, etc.

Several theoretical models for programmable matter have been proposed, ranging from DNA self-assembly systems,
 (e.g.,~\cite{xPa14}) 
to metamorphic robots,
 (e.g.,~\cite{WaWA04}),
 to  nature-inspired synthetic
insects and micro-organisms
  (e.g.,~\cite{DeGRSBRS15,xDoFRNV+16}).
Among them, the  {\em geometric Amoebot} model~\cite{xCaDRR16,xDaDGP+17,xDaGPR+17,DeDGRSS15,DeDGRSS14,xDerGMR+17}
is of  particular and immediate interest  from the distributed computing viewpoint.
In fact, in this model (introduced in ~\cite{DeGRSBRS15})
programmable matter  is viewed as a swarm of decentralized autonomous self-organizing  
entities (also called {\em particles}, operating on an  hexagonal tessellation of the plane. 
These  particles   have simple computational capabilities (they are finite-state machines), 
 strictly local interaction and communication capabilities (only with  particles located in neighboring nodes of the hexagonal grid), 
 and limited motorial capabilities (they can move only to   empty neighboring nodes); 
 furthermore, time is divided into round, and their activation at each round is controlled by an adversarial (but  fair) scheduler;
 the scheduler is  said to be sequential, fully synchronous, and arbitrary (or semi-synchronous) depending on  whether it activates at each round only one particles, all particles, or an arbitrary subset of them, respectively.
 A characteristic  feature of the Amoebot model is that, at any round,  a particle  can be  
 {\em contracted} (occupying   one node)  or  {\em expanded}  (occupying  two adjacent  nodes);
 it is through expansions and contractions  that particles move on the grid.
%
In this model, the research focus  has been on applications such as  {\em coating}~\cite{xDaDGP+17,xDerGMR+17}, {\em gathering}~\cite{xCaDRR16},  and   {\em shape formation}~\cite{DeDGRSS15,DeDGRSS14,DeGRSBRS15,DeGRSBRS16,noi2017}. 
The shape formation   (or pattern formation) problem  is  prototypical   for   systems of self-organizing  programmable particles, and particular attention has been given to
special basic shapes like the {\em line} \cite{DeGRSBRS15,DeGRSBRS16,noi2017}, that is used as a tool for the construction of more complex configurations.
%

A common feature of  the existing studies on  these programmable particles  is that all the particles are assumed  to be  fully operational at all times;
that is, faults have never been considered.
 In this paper we address   the presence of faulty particles.

We consider a connected shape of $n$ particles of which $f$ are faulty, and the faults are {\em crashes}.
We are interested in  the problem of the non-faulty particles  efficiently re-configuring  themselves  so to form
the same shape without including any faulty particles. We call this problem {\em shape recovery}, and is  a  basic task of   {self-reconstruction}/{self-repair}  
for  a prescribed shape.

In this paper we study this problem when the  shape the particles form is the {\em line}, hence the problem becomes that of
{\em line recovery}. Solving this  problem requires formulating a set of rules (the algorithm) that
will  allow the non-faulty entities to form the  line within finite time, regardless of the
initial distribution and number of faults and of the local orientations of the non-faulty entities.
Unfortunately this task, as formulated,   is actually  {\em unsolvable},
  even with a fully synchronous scheduler.
In fact, there are initial configurations where unbreakable symmetries make it impossible
to form a single line. 
We thus require
that  either one or two lines of equal size  be formed, depending on the symmetry level of
the  initial configuration.
This   problem  has been studied in \cite{fun2017} in a different computational setting, 
and solved in the case of a square grid assuming a fully synchronous scheduler. 

In this paper  we solve the  {\em line recovery} problem in the Amoebot model under a semi-sysnchronous adversarial scheduler.
We present a  {line recovery} algorithm  allowing $n-f>4$ non-faulty particles  without chirality to correctly  form either a single line or two lines of equal size, 
regardless of the position of the faulty particles and of the  number of non-faulty ones activated at each round.



\section{Model}

We consider the space to be an infinite unoriented anonymous triangular grid $G(V,E)$, where the nodes in $V$ are all equal and edges are bidirectional (see Figure~\ref{fig-grid1}). 
In the system there is a set $P$ of $n$ particles, initially located at distinct positions in $G$.
A subset $F \subset P$ of the particles, with $|F|=f$, is {\em faulty}: a faulty particle does not move or communicate with other particles. The other particles are said to be {\em correct}; the subset of correct particles is denoted by $C = P \setminus F$.

\begin{figure}
\begin{center}
\includegraphics[scale=1]{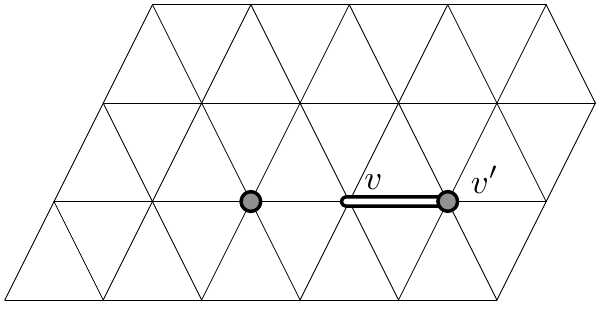}
\end{center}
\caption{A fragment of a triangular grid with two particles(grey circles). One particle is  expanded, with the head in $v'$ and the tail in $v$;
 the other is contracted.\label{fig-grid1}}          

\end{figure}

A particle $p$ assigns to each incident edge a distinct port number; this numbering is local and we do not assume that the particles agree on a common clockwise direction. Two particles that are neighbours in $G$ form a {\em bond}. Each particle has a shared constant size memory associated to each of its local ports, that can be read and written also by a neighbour particle. Moreover, each particle has a constant size memory used to store its state. 

To ease the writing, in our algorithm we will use a message passing terminology, in which, when a particle sends a message to a neighbour, it writes on the shared memory of the receiving neighbour. Symmetrically, if a particle receives a message, it will find it in its shared memory. 

The system works in rounds, and particles are {\em activated} by an external semi synchronous schedule: at each round $r$ the scheduler selects a subset of correct particles, $C_r \subseteq C$, and it activates them; at the same round, the particles in $C \setminus C_r$ are inactive. The scheduler is fair, in the sense that it has to activate each particle infinitely often.  
A particle $p$ moves by a sequence of {\em expansions} and {\em contractions}: a contracted particle occupies a single node $v \in V$, while an expanded one occupies two neighbours node nodes. Initially, each particle occupies exactly a single node; i.e., all particles are contracted. During the execution of the algorithm, a correct contracted particle $p$ that is in $v$ might expand: after the expansion, $p$ will occupy two nodes: $v$ and the neighbour node $w$ where it expanded to. We will say that node $w$ is the {\em head} of the particle, and $v$ is the {\em tail}. Particle $p$ always knows which node is its head and which one is its tail. If a particle is expanded, it can contract back in either  tail or head node (if a particle is contracted, node and tail are the same). We assume that a particle $q$ that is a neighbour of $p$ knows if $p$ is contracted or expanded, and it knows if it is bonded with the tail or the head of $p$. Also, particles are endowed with a failure detector, that takes as input a local port number and returns true if its neighbour (if any) is in $F$.

Upon activation at round $r$, a particle $p$ executes the following operations: \\
{\bf Look}: It reads the shared memories of its local ports, the shared memories of the local ports of its neighbours (if any), and the output of the failure detector on each port.\\
{\bf Compute}:  Using these information and its local state, it performs some local computations. Then, it updates its internal memory and it possibly writes on the shared memories of its neighbours. As an outcome of the Compute operation, $p$ can either decide to stay still or to Move.\\
{\bf Move}: If in the Compute operation $p$ decides to move, then it can either expand to an occupied neighbour location, if contracted at round $r$, or contracts towards its head or tail, if expanded. Moreover, it can perform a special operation called {\em handover}, in which it forces the movement of a correct neighbour particle $q$: if $q$ is expanded and $p$ is contracted, then $p$ forces the contraction of $q$, by pushing $q$ towards its head/tail occupying the tail/head of $q$; otherwise, if $q$ is contracted and $p$ is expanded, then $p$   contracts towards its head/tail forcing the expansion of $q$ in the tail/head of $p$.

Since the scheduler can activate more than one particles in the same round, it is crucial to specify what happens in case of conflicting operations executed by different particles: \\
(i)  If two or more particles try to expand in the same node $v$, then only one succeeds, and the decision depends on the scheduler. The particle that fails to move, will be aware of this at the next activation, by realising that it is still contracted. \\
(ii) If two or more   particles try to execute an handover with a particle $p$ and $p$ is moving, then only one will succeed (which one depends on the scheduler). The ones that fail to move  will be aware of this at the next activation, by realizing that they have not moved.\\
(iii)  If a particle $p$ tries to execute an handover with a particle $q$, and $q$ is moving, then the handover succeeds if and only if also $q$ is executing the same handover operation with $p$. Otherwise, $q$ moves, $p$ fails the handover, and $p$ will be aware of this at the next activation, by realizing that it has not moved.

Given a set of particles $P$, we say that they are on a {\em compact straight line} if they are on a  {\em straight line} and the subgraph induced by their positions is connected. Initially, at round $r=0$, all particles (both correct and faulty) are positioned on a compact straight line (the \vir{initial line}). In the following, we will assume that $n-f \geq 5$. The \algo problem is solved at round $r^{*}$ if, for any round $r \geq r^{*}$, the particles in $C$ either form a straight line with an unique leader particle, or they form two straight lines of equal size each one with its own leader particle.

\section{Line Recovery Algorithm}

\subsection{Overall Description}
In the following, we will assume that the particles can exchange fixed-size messages (this can be easily simulated in our model). Also,
if not otherwise specified, the variable $dir$ of a particle $p$ stores the movement's direction of $p$; that is, it stores the port number where $p$ intends to move; when no ambiguity arises, we will use the expression ``direction of $p$'' to indicate the content of $dir$. Similarly, the content of variable $pre$ stores the location of $p$ in the previous round; again, when no ambiguity arises, we will use the expression ``previous location of $p$'' to denote the content of this variable. 
Moreover, we will say that $p$ is \emph{pointing at} a particle $p'$ if $p$ and $p'$ are neighbors, and the direction $dir$ of $p$ is toward the location occupied by $p'$. Finally, when a particle $p$ changes state from $s$, we will say that $p$ {\em becomes} $s$.

Let $L_0$ be the line where the particles are placed at the beginning, and $L_1$, $L_{-1}$ be the two lines adjacent to $L_0$.
The overall idea to solve the problem is as follows: first, the particles, try to elect either one or two leaders. The election is achieved by having few selected particles move on $L_1$ and $L_{-1}$; also, during this movements, no gap of size larger than $3$ is ever created: this is a crucial invariant to keep, to correctly understand whether a particle is on one of the two extremes of the line. The role of the leader(s) is to start the construction of the final straight line. If there is only one leader,  it will perform a complete tour around $L_0$ (i.e., moving on $L_1$ and $L_{-1}$): during this tours it collects all correct particles, that will follow the leader(s), by attaching to the line they are building. If there are two leaders, then one must be placed on $L_{1}$ and the other on $L_{-1}$. Each one will build a line of correct particles, and then they will compare the length of such lines. If the lines are not equal the symmetry is broken and a single line of correct particles will be formed; otherwise, two equal sized lines will be formed.

In more details, the \algo algorithm is divided in seven sub-algorithms: {\sf Fault Checking}, {\sf Explorer Creation}, {\sf Candidate Creation}, {\sf Candidate Checking}, {\sf Unique Leader} and {\sf  Opposite Sides}.

The algorithm starts by checking whether  there are no \faulties: in this case, all particles in $C$ are already forming a compact line.
 This scenario is detected during the {\sf Fault Checking} sub-algorithm, started  by the particles who occupy the extreme positions of the starting configuration, i.e., by the two particles having  only one neighbor, these particles get state \marker. These two extreme particles send a special message inside the line: if the two messages meet, there are no \faulties.

Should there be  \faulties, the second sub-algorithm (the {\sf Explorer Creation}) is performed, started by all the \settler{{\em s}} who have  a \faulty  neighbor. 
In this sub-algorithm some \settler{{\em s}} become  \explorer{{\em s}} and move out of the line (either on $L_1$ or on $L_{-1}$). The selection of the \explorer{{\em s}} is made in such a way that their movement does not  create ``gaps'' of more than two consecutive empty positions anywhere in the original line (this property is crucial to detect the end of the line in subsequent  sub-algorithms).

\subsection{Sub-Algorithms}

\algblockdefx[Upon]{Upon}{End}
  [1][Unknown]{{\bf Upon Activation in State} {\sf #1} {\bf do:}}%
    {{\bf End}}

\begin{figure*}
\begin{framed}
\footnotesize
\begin{algorithmic}[1]

\Upon[{\em Init}]
  \State Set Line $L_{0}$=getLineDirectionFromActivatedInitAndFaultyNeighbours()
 \If{$\exists p \in Ports | msg.switchtoslave  \in p$}
   \State Set State \slave
    \ElsIf{ $ (\exists ! p \in Neighbours |  (p.state  = \textsf{Init} \lor p.state  = \textsf{\inited} \lor faulty(p))$} \label{init:linebootstrap}   

    \State Set $flag.linetail$
    \State $lineparity=1$
    \State $send(p,msg.coin)$
        \State Set State {\sf \marker}
    \Else
    \State Set State {\sf \inited}
    \EndIf
\End
\\

\Upon[\marker]
	\If{$\exists p \in Ports | msg.markertoleader \in p$}
	   \State Set State {\sf \exodusleader{{\em s}}}
	\ElsIf {$\exists p \in Ports | msg.asktobecandidate \in p \land \neg flag.candidate$} \label{marker:if1}
	\State Set $flag.candidate$
	\State $send(p,msg.candidate)$
	
		\ElsIf  {$\exists port \in Ports | msg.switchtoleader \in port$} 

		 	\State Set direction and flags from $msg.switchtoleader$
		\State Set State {\chosen} 
		
	\ElsIf  {$\exists p \in Neighbours | p= \text{\probe}  \land contracted(p) \land \exists p' \in Neighbours | p' = \text{\collector{\em .counting}}$}  
		\If{$\nexists msg.seen \in p$}
		\State $send(p,msg.seen)$
		\State $send(p',msg.other)$
		\EndIf
		
	\ElsIf  {$\exists p \in Neighbours | p= \text{\collector{\em .done}}  \land contracted(p) \land \exists p' \in Neighbours | p' = \text{\collector{\em .counting}}$}  \label{collectorwinner}
		\State $send(p',msg.winner)$
	\ElsIf  {$\exists p \in Neighbours | p= \text{\collector{\em .done}}  \land contracted(p) \land \exists msg.done  \in opposite(p)$} \label{collectoreven}
		\State $send(p',msg.even)$
		\State Set State {\recruiter}

	\EndIf
\End
\end{algorithmic}
\end{framed}
\caption{Algorithm for  {\em Init}, \marker and {\em \inited} -- Part One \label{algorithm:activated}}
\end{figure*}

\algblockdefx[Upon]{Upon}{End}
  [1][Unknown]{{\bf Upon Activation in State} {\sf #1} {\bf do:}}%
    {{\bf End}}

\begin{figure*}
\begin{framed}
\footnotesize
\begin{algorithmic}[1]
\setcounter{ALG@line}{31}
\Upon[{\em \inited}]

 	\If{$\exists p \in Ports | msg.switchtoslave  \in p$}
   		 \State Set State \slave
  		  \State End Cycle

	 \ElsIf{$\exists p \in Ports | msg.coin \in p$}     
	 	\If{$\exists neighbour \in  opposite(p) \land neighbour.linetail$} \label{activated:newtail} 
    		\State $send(neighbour,msg.newtail)$
		\State $parent=neighbour$
    		\State Set $flag.linetail$
		\State $lineparity=neighbour.lineparity++$
   		 \State $send(p,msg.coin)$
		 \EndIf \label{activated:newtailend} 
		\If{$\exists neigh \in  opposite(p) \land neigh \neq \text{{\em Init}} \land \nexists msg.coin \in neigh.port$} \label{activated:msgcoin}  
    		 
		 \State $send(neigh,msg.coin)$
		\Else
		\State Add $msg.coin$ to $p$                   \label{activated:msgcoinend}
    	 \EndIf
	  \ElsIf{$flag.linetail$}
	   	   \If{$\exists p \in Ports | msg.newtail \in p$}                    \label{activated:unset}

		\State Unset $flag.linetail$ \label{activated1}
		\EndIf
 	 \If{$ \exists neighbour \in Neighbours \land neighbour \neq parent \land neighbour.flag.linetail $}  \label{activated:parition}

		\State $p =  opposite(neighbour)$
		\If{$lineparity=0 \lor lineparity=neighbour.lineparity$}
	 	\State $send(p,msg.markertoleader)$    \label{activated:paritionend}
		 \EndIf
  	 \EndIf
	
	 \EndIf
	   \If{$\exists p \in Ports | msg.markertoleader \in p$}
		\State $p =  opposite(neighbour)$
		\State $send(p,msg.markertoleader)$
		\EndIf
   \If{ $ (\exists  p \in L_{0}  |  faulty(p)) \land (\exists p' \in L_{0} | (p' \neq p \land p' \neq \text{{\em Init}} \land (\nexists port \in Ports | msg.notified \in port) ) $} \label{activated:ifpreexp}
     \State $send(p',msg.notify)$
    \State Set State {\sf pre.\explorer}
    \EndIf
    \If{ $ (\exists! port \in Ports | msg.notify \in port) \land (\nexists p \in L_{0} | p = Init$) }\label{activated:notify}
    \State $send(opposite(port), msg.notified)$  
    \State Add $msg.notify$ to $port$  
        \State Set State {\sf notified}
    \EndIf
\End

\end{algorithmic}
\end{framed}
\caption{Algorithm for  {\em Init}, \marker and {\em \inited} -- Part Two  \label{algorithm:activatedii}}
\end{figure*}

\subriteDescription{Fault Checking}
In the {\sf Fault Checking} sub-algorithm (reported in Figures~\ref{algorithm:activated} and~\ref{algorithm:activatedii}) the particles  detect if there are no \faulties. If so, they elect either one or two \exodusleader{{\em s}}. In case two \exodusleader{{\em s}} are elected, two lines having exactly the same size will be formed.

In the following, we will use the following convention: in the pseudo-code, when there is an {\bf If} statement that checks the presence of a particular message (i.e., Line~\ref{activated:unset}), and the guard of the {\bf If} statement is {\bf true}, then the message is immediately deleted from the memory (i.e., when Line~\ref{activated1} is executed on particle $p$, the $msg.newtail$ of the statement of~\ref{activated:unset}  is deleted from the memory of $p$).

This routine is started by the particles who occupy the extreme positions of the starting configuration, i.e., by the two particles having  only one neighbor, called \marker. 
The general idea behind this sub-algorithm is as follows: each \marker generates a special message that travels towards the other \marker; if the two messages meet, than there are no \faulty particles. 
In particular, two $msg.coin$ messages are generated, one from each \marker, travelling in opposite directions; the other particles, when activated for the first time, switch from {\em Init} to {\em \inited} state, and save the direction of the initial line $L_0$ in its local memory. Without loss of generality, let $m_l$ and $m_r$ be the leftmost and rightmost \marker, respectively; also, let us denote by $S_l$ and $S_r$ the {\em segments} of $m_l$ and $m_r$, respectively: at the beginning,  $S_l$ and $S_r$ contain only $m_l$ and $m_r$, respectively (i.e., each \marker is both the start and the end of its own segment). When the $msg.coin$ sent by $m_l$ reaches the end $S_r$, $S_r$ expands of one unit towards $m_l$, and the $msg.coin$ is sent back to $m_l$; symmetrically, when the $msg.coin$ sent by $m_r$ reaches the end $S_l$, $S_l$ expands of one unit towards $m_r$, and the $msg.coin$ is sent back to $m_r$. During this expansion process, the end of each segment stores the parity of the length of the segment it belongs to. This process is iterated until, if there are no \faulty particles, the ends of the two segments will be neighbors: when this occurs, if the two segments have the same size, two leaders will be elected; otherwise their sizes differ by at most one unit and a unique leader is elected, that is the \marker of the segment with length of parity 0.
An example run is in Figure \ref{xf:shape}.

\begin{figure*}[tbh]
\begin{center}
  \subfloat[The extremes of the line start the procedure to partition the line. Initially, each extreme is also a tail, represented by the marked circle, and they have counter $1$, represented by the light grey color.]{
  \includegraphics[scale=0.5]{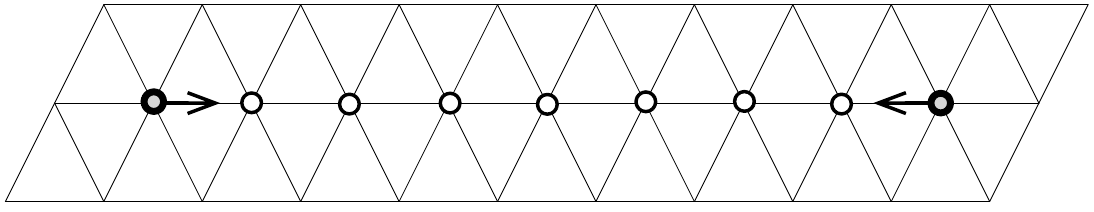}
  }
    \qquad\quad
  \subfloat[The $msg.coin$, represented by an arrow, moves from on tail to the other.]{
  \includegraphics[scale=0.55]{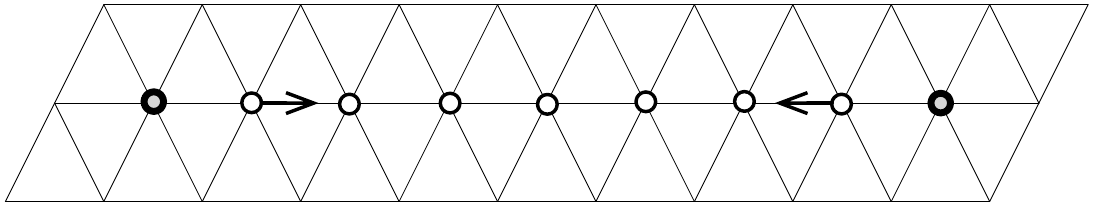}
  }
   \\
      \subfloat[The $msg.coin$ reaches a tail.]{
  \includegraphics[scale=0.55]{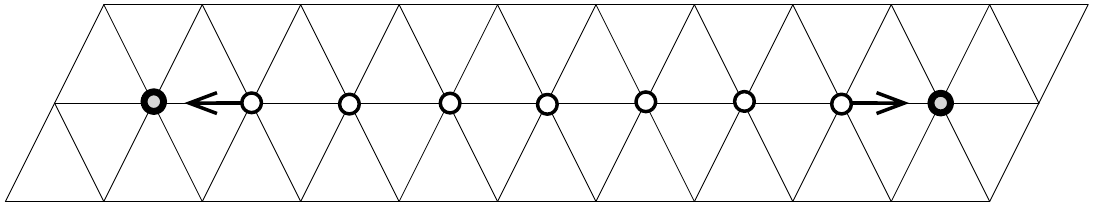}
  }
   \qquad\quad
  \subfloat[New tails and new $msg.coins$ are created. Each segment of the line increases by one, we can see that the new tail have value $0$ for the counter, represented by the dark grey color.]{
  \includegraphics[scale=0.55]{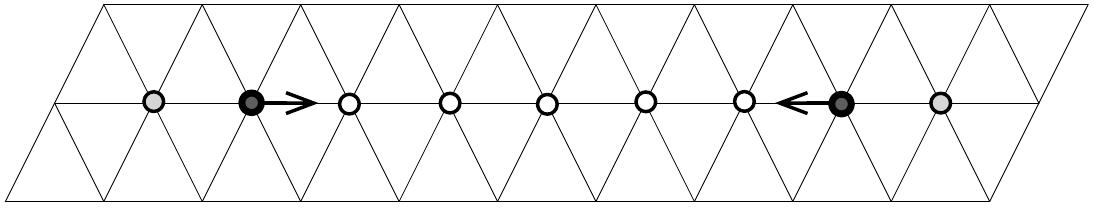}
  }
     \\
      \subfloat[The $msg.coins$ travel with different speeds, due to the semi-synchronous activations.]{
  \includegraphics[scale=0.55]{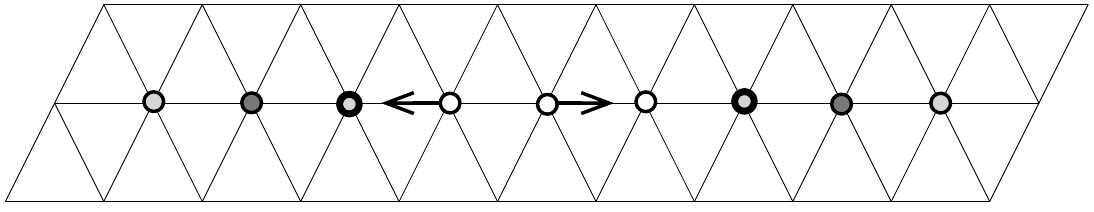}
  }
   \qquad\quad
  \subfloat[Two $msg.coins$ travel towards the same tail.]{
  \includegraphics[scale=0.55]{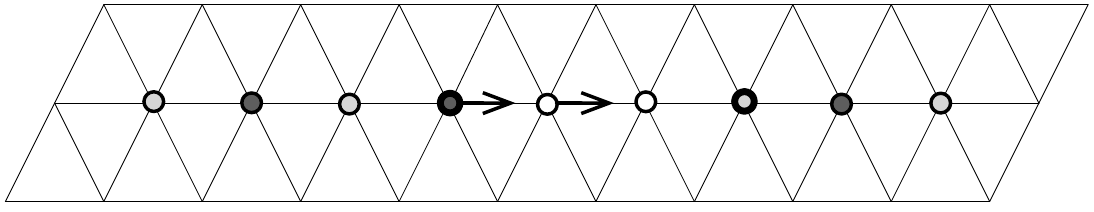}
  }
       \\
      \subfloat[A segment grows of two units.]{
  \includegraphics[scale=0.55]{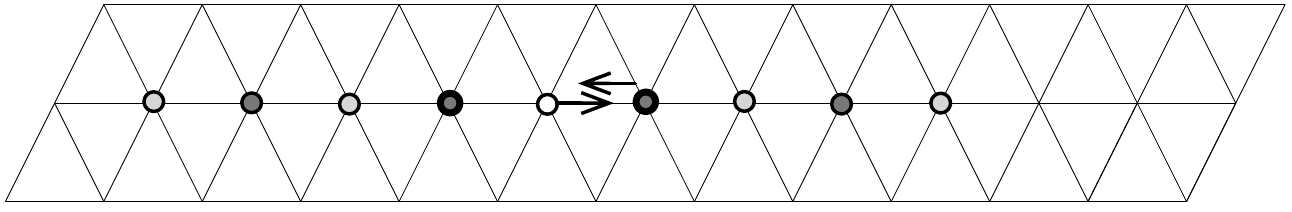}
  }
   \qquad\quad
  \subfloat[Two tails are neighbours. They have different parity so a single leader can be elected]{
  \includegraphics[scale=0.55]{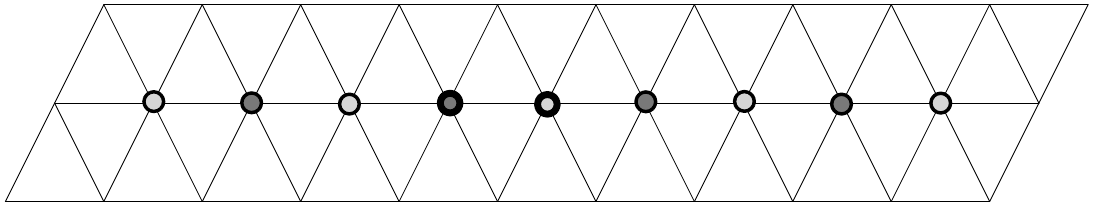}
  }
\end{center}
\caption{Run of the algorithm where the Fault Check subphase elects an unique leader.}
\label{xf:shape}
\end{figure*}

\begin{lemma}\label{lm:0}
By executing procedure {\sf Fault Checking}, the size of the two segments differ by at most $1$ unit.
\end{lemma}
\begin{proof}
We first show that the procedure that expands the two tails forces the size of the two segments to differ by at most $1$.
In fact, the size of a segment is the number of nodes between a \marker and its respective tail. The proof is by induction on the number of rounds. 
Let $r$ be the first round at which a \marker is created. At the end of round $r$, if the \marker exists, each segment has size $1$; otherwise the segments have both size $0$. Therefore the claim is verified. 
 
Let us suppose that the claim holds until round $r+k$. At round $r+k+1$ each segment may increase its size of at most one unit; if at round $r+k$ the difference between the size of the two segments was $0$, then the statement is verified. Otherwise, if the difference was $1$, then we need to show that the bigger segment, say $s$, does not increase until the smaller one, say $s'$, increases by one unit.

A segment, to increase its size, has to consume a message $msg.coin$. This message can only be produced by the tail of the other segment, and it is produced only when the segment grows by one unit. Notice that there is no $msg.coin$ travelling from the tail of the smaller segment $s'$ to $s$. Otherwise, the difference in size between the two segments would have been more than $1$ at a round $r' < r+k$: $s'$ would have grown of one unit, and still be smaller than $s$.
\end{proof}

\begin{theorem}\label{thf:0}
If $f=0$, then the procedure {\sf Fault Checking}  correctly solves the \algo problem.
\end{theorem} 
\begin{proof} It is immediate to see that, as long as the two tails are not neighbours, a segment will eventually grow. We distinguish two cases:

\begin{enumerate}
\item If the line has an even number of particles, then  {\sf Fault Checking}  divides the initial line in two segments of equal size. By contradiction, let us assume  that
one of the two segments is greater than the other, and that two tails are touching. Being the line even, this can only be possible if one of the two segments is two units or more longer than the other one: by Lemma~\ref{lm:0}, this is not possible, and two \markers are both elected as leader.

\item If the line has an odd number of particles, and the two tails are touching, we have that only one of the segment has an odd size. In fact, they cannot have both odd size, otherwise the initial line would have an even number of particles, having a contradiction. Therefore, the two tails have different parities, and by executing {\sf Fault Checking} only one leader is elected. 
\end{enumerate}

In both cases, the theorem follows.
\end{proof}

\subriteDescription{Explorer Creation}
The sub-algorithm {\sf Explorer Creation} is used to bootstrap the other sub-algorithms: its execution is started from sub-algorithm {\sf Fault Checking} if in the system there is at least one \faulty particle. The main purpose of this sub-algorithm is to select, among the correct particles,  at least three  \explorers, who will
 move out of line  $L_0$ without creating empty ``gaps'' of more than two consecutive positions. This is done as follows. 
 
 \begin{figure}[H]
\begin{framed}
\footnotesize
\begin{algorithmic}[1]
\Upon[{\em pre.}\explorer]
	 \If{$\exists port \in Ports | msg.switchtoslave  \in port$}
    	\State Set State \slave
    	\State End Cycle
 	\EndIf
	\State $l=getNeighbourOutside(L_0)$ \label{movepreexp1}
	\State $direction=left$
	\State $move(l)$
	\State Set State {\sf \explorer}  \label{movepreexp2}
\End
\\
\Upon[{\em notified}]
	 \If{$\exists port \in Ports | msg.switchtoslave  \in port$}
   	 \State Set State \slave
    	\State End Cycle
	 \EndIf
	 \State $cond_1:= \nexists p \in Ports | msg.notified \in p$
	 \State $cond_2 := \nexists p \in opposite(port(msg.notify)) |   p = pre.\text{\explorer}$
	 \State $cond_3 :=  \nexists  msg.notify \in opposite(port(msg.notify)$
	    \If{ $ (cond_1 \land cond_2 \land cond_3)$} \label{notified:if1}
	        \State Set State {\sf pre.\explorer}
  	  \EndIf
\End
\end{algorithmic}
\end{framed}
\caption{Algorithm for {\em pre.}\explorer and {\em notified} \label{algorithm:preexplorer}}
\end{figure}

If a particle in {\sf \inited} state (from {\sf Fault Checking})  has a  \faulty neighbour, then it becomes {\em pre.}\explorer, and it notifies this decision to any non \faulty neighbouring  particle.  If the neighbour is correct but it is still in the {\em Init} state, it waits (see Line~\ref{activated:ifpreexp} of Figure~\ref{algorithm:activatedii}).

A  {\sf \inited} particle that, upon activation, finds such a notification message, becomes \emph{notified} (Line~\ref{activated:notify} of Figure~\ref{algorithm:activatedii}).
A \emph{notified} particle that is not a \marker becomes  a {\em pre.}\explorer if, on the opposite port to the one containing the notification, there is no neighbour that will become either \emph{notified} or {\em pre.}\explorer; otherwise, it stays in the \emph{notified} state. 
Note that, when a particle changes state to \emph{notified}, it  sends a message to its other neighbour, to avoid that it also becomes a {\em pre.}\explorer.

If a {\em pre.}\explorer is activated, it becomes an \explorer, it moves outside $L_{0}$ on
$L_{j}$ with $j \in \{1,-1\}$, and it picks as direction on $L_{j}$ the left one, according to its chirality (see Lines~\ref{movepreexp1}-\ref{movepreexp2} of Figure~\ref{algorithm:preexplorer}). The direction is stored in the local variable $direction$, that is always
pointing to some location on   line $L_{j}$ on the left of the current one according to the handedness of the particle. 
Notice that the particle might fail to leave $L_{0}$, because of collisions with other particles: to handle this case, an \explorer that finds itself on line $L_{0}$ tries to expand to go
outside $L_0$.

Let a {\em sequence of particles} be a set of consecutive particles; also, let a {\em gap} be the maximum number of empty locations between two particles (either correct or \faulty) on line $L_0$.
An example run of Explorers creation is in Figure \ref{xf:shape1}.

\begin{figure*}[ht]
\begin{center}
  \subfloat[Initial coinfiguration: we have a sequence of two correct particle,  a sequence of $3$ faulty particles, a sequence of $3$ correct particles and a faulty particle at the end of the line.]{
  \includegraphics[scale=0.55]{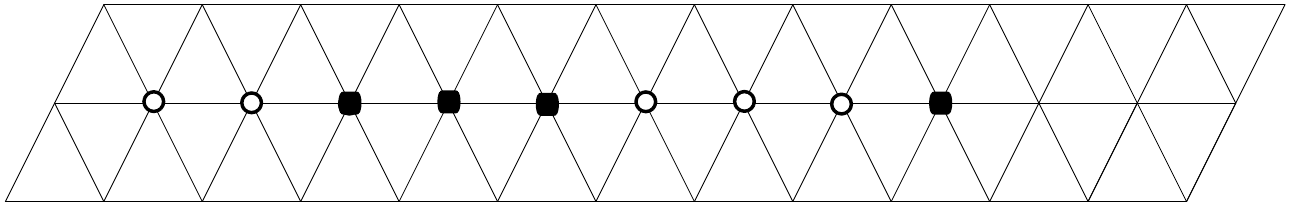}
  }
    \qquad\quad
  \subfloat[The left endpoint of the line becomes a \marker, red star. The pre-\explorers are the yellow crosses and they are neighbours of faulty particles. The notified particle is the cyan rectangle and it is sending a notified msg to the white particle, that is still in the \inited state. This implies that the white particle will never become an \explorer. ]{
  \includegraphics[scale=0.55]{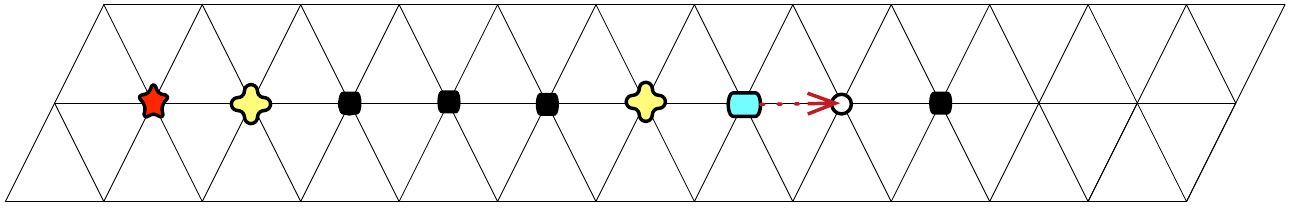}
  }
   \\
  \subfloat[The pre-\explorers become \explorers and the notified becomes a pre-\explorer.]{
  \includegraphics[scale=0.55]{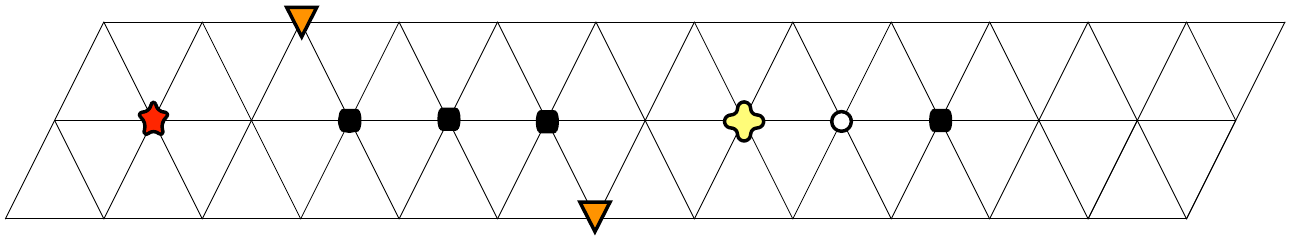}
  }
    \qquad\quad
  \subfloat[All the \explorers have been created.  ]{
  \includegraphics[scale=0.55]{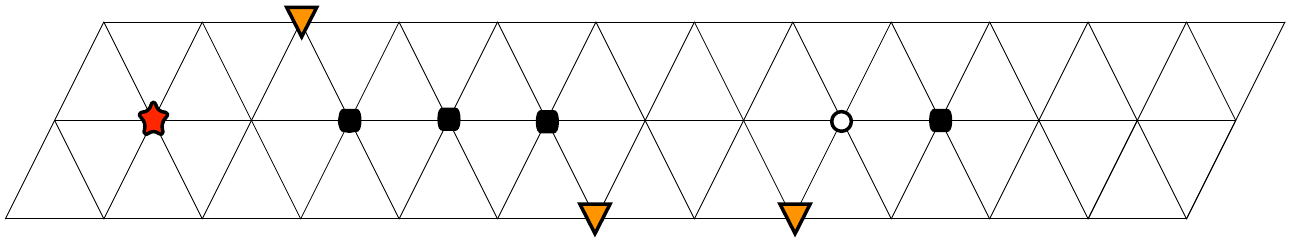}
  }
   
\end{center}
\caption{Explorer Creation, example run.}
\label{xf:shape1}
\end{figure*}

\begin{lemma}\label{lemma:explorercreation}
Sub-algorithm {\sf Explorer Creation} never creates a gap of size $3$. Moreover, for each sequence of correct particles of size greater or equal $3$,
at least two pre.\explorer will be created. 
\end{lemma}
\begin{proof}
The proof considers runs of correct particles of different size. First, it is easy to see that for any run of size $2$ we cannot create a gap of size $3$.

For a run of size $3$, let $p_1, p_2, p_3$ be the placement of the three correct particles.
First, notice that at least two {\em pre.}\explorer{\em s} will be created: either the two immediate neighbours of a \faulty particle, or an immediate neighbour of a \faulty
and a {\em notified}.
Without loss of generality, let $r$ be the first round at which particle $p_1$ becomes a {\em pre.}\explorer. Also, let $r' \geq r$ be the first round at which the particle $p_2$ wakes up after $r$. If also $p_3$ becomes a {\em pre.}\explorer before round $r'$, then $p_2$ receives two notify messages.
Therefore, $p_2$ will not became {\em notified}; hence, it cannot become {\em pre.}\explorer. 
Otherwise, if $p_3$ becomes {\em pre.}\explorer at round $r'$, $p_2$ receives an additional notify message at round $r'$; so, when $p_2$ is activated after becoming {\em notified}, the
predicate that leads the particle to be a {\em pre.}\explorer cannot be verified ($\exists! port \in Ports | msg.notify \in port$, see Line~\ref{notified:if1} of Figure~\ref{algorithm:preexplorer}), hence $p_2$ cannot become {\em pre.}\explorer.
Finally, by predicate $(\exists p' \in L_{0} | (p' \neq p \land p' \neq \text{{\em Init}} \land (\nexists port \in Ports | msg.notified \in port) )$, $p_3$ cannot become a {\em pre.}\explorer after round $r'$. 

For a run of size $4$, let $f, p_1, p_2, p_3, p_4, f$ be the placement of the particles, with $f$ the faulty processes, and let us examine the round $r$ at which $p_2$ becomes {\em notified}.
Let us suppose, w.l.o.g., that $r$ is the first round at which a particle becomes {\em notified}.  If at $r$ particle $p_3$ does not become {\em notified},
then notice that $p_3$ cannot become a {\em pre.}\explorer either (by predicate $(\nexists port \in Ports | msg.notified \in port)$), so a gap of size $3$ cannot be created. If at $r$ also $p_3$ becomes {\em notified}, then,
for the same predicate, both $p_3$ and $p_2$ cannot become {\em pre.}\explorer. 
Also, notice that at least two {\em pre.}\explorer{\em s} will be created: the two immediate neighbours of a \faulty.
Notice, that if the run is $f,p_1,p_2,p_3,p_4$, with $p_4$  the last process of the initial line $L_0$, then both $p_1,p_2$  become {\em pre.}\explorer; thus 
also in this case at least two {\em pre.}\explorer{\em s} are created. 

In the general case of a run of size greater than $4$, notice that only the immediate neighbour of a \faulty particle and its neighbour can become {\em pre.}\explorer. So, also in this case, a gap of size $3$ cannot be created. It thus follows that also in this case at least $2$ {\em pre.}\explorer{\em s} will be created, and the lemma follows. 
\end{proof}

\begin{observation}\label{observation:makermarker}
If at round $r=0$, there is a correct particle $p$ at the end of the initial line, then this particle will become a \marker.
\end{observation}
\begin{proof}
If this particle has a non-faulty neighbour on the line, then such neighbour will not move until $p$ is in the {\em Init} state.
When $p$ is activated, it will verify that it has only one neighbour in state {\sf \inited}, and it will thus become a \marker.
The same occurs if the neighbour on the line is \faulty. 
\end{proof}

\begin{lemma}\label{lemma:numberexplorer}
There exists a round $r$ in which there is either: (1) a \marker and at least two {\em pre.}\explorer{\em s}, or (2) at least three {\em pre.}\explorer{\em s}.
\end{lemma}
\begin{proof}
Let us first examine the case in which, at round $r=0$ we have a run of at least $3$ correct particles near the end of the line. By Lemma~\ref{lemma:explorercreation} and Observation~\ref{observation:makermarker}, the lemma follows.

Let us now examine the case in which the run has at least $2$ correct particles near the end of the line. By Observation~\ref{observation:makermarker}, there is at least one \marker, and the neighbour of the \marker will become a {\em pre.}\explorer.

By hypothesis, $n-f \geq 5$; therefore, if there is a run of $3$ correct particles, by Lemma~\ref{lemma:explorercreation}, at least two {\em pre.}\explorer{\em s} will be created, and the lemma holds. Otherwise, notice that if the run has only $2$ correct particles, again at least one {\em pre.}\explorer will be created, and the lemma holds as well. 
The lemma holds also if there are two different runs, each of size one. 

Finally, let us examine the case in which, at round $r=0$, there is at least $1$ correct particle near the end of the line. By Observation~\ref{observation:makermarker}, there is at least one \marker. 
The remaining correct particles form either a run of $3$  particles; or a run of two correct particles near the end of the line and another run of at least one correct particle inside the line; or two runs of size less than $3$ inside the line. In all these cases the lemma holds. 

The last case to consider is when there is no correct particle near the end of the line at round $r=0$: in this case, if the correct particles are all in one run, we have $4$ {\em pre.}\explorer{\em s}. Otherwise, if there is a run of size $3$, then we still have at least $3$ {\em pre.}\explorer{\em s}: two created in the run of size $3$, and the third one among the remaining correct particles. 
If there are two runs of size $2$, they will both create $2$ {\em pre.}\explorer{\em s}. The lemma still holds if there are five runs of size one, or a run of size two and three runs of size one.
\end{proof}

\subriteDescription{Candidate Creation}  
\begin{figure}
\begin{framed}
\footnotesize
\begin{algorithmic}[1]
\Upon[{\em pre.}\marker]	\label{explorer:premarker}

	\If{$myself.expanded$}
	\State $contracToHead()$
	\State Set State \marker
	\Else
	\State Set State \explorer
	\EndIf
	\color{black} \label{explorer:premarker1}
\End
\\
\Upon[\explorer]
		 \If{$\exists port \in Ports | msg.switchtoslave  \in port$}
    	\State Set State \slave
    	\State End Cycle
    \ElsIf{ $ myself \in L_{0}$}
	   	\State $l=getNeighbourOutside(L_0)$
		\State $direction=left$
		\State $move(l)$
	   \ElsIf{$\exists p \in L_0 | p={\sf Init}$}
	\State End Cycle
	  \ElsIf{$\exists l_1,l_2,l_3 \in L_0 | empty(l_{1,2,3})=true \land l_1,l_2,l_3$ are contigous} \label{explorer:marker}
	  \State pre.\marker
	  \State $expand(l_3)$\label{explorer:marker2}
	    \ElsIf{$\exists p \in getLocation(direction) \land p=\text{\explorer} \land pointingAtMe(p) \land \nexists mgs.changedirection \in p$} \label{explorer:cd1}
	    \State $send(p,mgs.changedirection)$
	    \State $direction=opposite(direction)$
	    \ElsIf{$\exists p \in Ports | mgs.changedirection \in p \land direction=p$} \label{explorer:cd2}
	    \State $direction=opposite(direction)$
         \ElsIf{$\exists p \in L_0 \land p=\text{\marker} $} \label{explorer:stopcandidate}
	    \State $send(p,msg.asktobecandidate)$
	    \State Set State \slave
	    	\ElsIf {$\exists! p \in Port | msg.switchtocandidate \in p$}	 \label{explorer:stopc1}	
	\State Set direction and flags from $msg.switchtocandidate$
	\State Set State {\leader}
	    \ElsIf {$\exists p\in Port | msg.switchtocandidate \in p$} \label{explorer:tochosen}    
	    	\State Set direction and flags from $msg.switchtocandidate$
	    \State Set State {\chosen} \label{leader:case3}
	    	\ElsIf  {$\exists p \in Ports | msg.switchtoleader \in p$} 
		 	\State Set direction and flags from $msg.switchtoleader$
		\State Set State {\chosen} 
	    			\ElsIf  {$(\exists! p \in Ports | msg.switchopposer \in p ) $}
		\State Set direction and flags from $msg.switchopposer$
	\State Set State {\opposer} 
				\ElsIf  {$(\exists p \in Ports | msg.switchopposer \in p ) $}
		\State Set direction and flags from $msg.switchopposer$
\State Set State {\chosen} 
         \ElsIf{$\exists p \in L_0 \land (p=\text{\opposer} \lor p=\text{\chosen}) $} 
	    \State Set State \slave 
	    \ElsIf{$empty(getLocation(direction))$}
	    \State $move(getLocation(direction))$
	  \EndIf
\color{black}
\End
\end{algorithmic}
\end{framed}
\caption{{Explorer Algorithm} \label{algorithm:explorer}}
\end{figure}
This sub-algorithm, reported in Figure~\ref{algorithm:explorer}, is executed  when $f >0$ at round $r=0$. Its main purpose is to elect at most two \explorers; one of the two elected \explorers becomes
a \leader. 
In this sub-algorithm the \explorer{\em s} move along the line until they find either the end of the line,  which 
is detected by seeing three consecutive empty locations, or  a \marker. 
If an \explorer meets a particle in the {\em Init} state, it waits. 

The first \explorer that reaches an extreme of the line without \marker, becomes a \marker and stays there.
If  two \explorers try to become \marker on the same end of the line at the same time, only one will succeed (they will both try to move in the same location, see Lines~\ref{explorer:marker}-\ref{explorer:marker2}  and \ref{explorer:premarker}-\ref{explorer:premarker1} of Figure~\ref{algorithm:explorer}).
An \explorer communicates its direction of movement to other \explorer{s} by writing an appropriate flag in the shared memory of the port where the head is going to expand.
In the following, we will say that the \explorer is pointing in some direction. 

If two \explorers meet and they have opposing directions, since they cannot pass through each other, they simply switch directions.
When an \explorer switches direction, it sends a message to the other explorer, to ensure that it will also switch direction. This message is also used to ensure that a particle does not switch direction twice with the same particle; that is, the \explorer checks that the other is pointing at it, and that it has not a pending $msg.changedirection$ in the shared memory of the corresponding port (see Lines~\ref{explorer:cd1},\ref{explorer:cd2} of Figure~\ref{algorithm:explorer}).
If an \explorer finds its next location occupied, it waits. Depending on the initial configuration, either one, two, or no \markers are created during this procedure. 

An  \explorer who reaches a \marker, sends a message to the \marker asking to become a \leader; if the \marker accepts, then the \explorer becomes a \leader. If two  \explorers reach the same  \marker and both ask to become a candidate, then the  \marker  will answer affirmatively to only one of them (see Lines~\ref{explorer:stopcandidate},\ref{explorer:stopc1} of Figure~\ref{algorithm:explorer}, and  Line~\ref{marker:if1}  of Figure~\ref{algorithm:activated}).  

\begin{lemma} \label{lemma:markerexplorer}
Starting from any initial configuration where $f>0$, there exists a round $r$ in which there is at least one \marker $p$ that signals the end of the line, and
an \explorer $p'$ moving towards $p$. 
\end{lemma}
\begin{proof}
Let us first examine the case in which at round $r=0$ both endpoints, $p_1$ and $p_2$ are in $C$. In this case, it is easy to see that they will both be activated resulting 
in having two \marker{s} at the end of the line. Moreover, by Lemma~\ref{lemma:numberexplorer}, we have at least three pre.\explorers; therefore, there will be one \explorer moving towards one of the \marker.
 
In case at round $r=0$ only one of the endpoints, say $p_1$, is in $C$, it will be eventually activated becoming a \marker.
By Lemma~\ref{lemma:numberexplorer}, we have that at least 2 pre.\explorer will be created: if at least one of them has the direction of $p_1$, then the lemma follows.
Otherwise (they have the same direction), one of them will either move towards a \marker, or towards the end of the line with a \faulty particle, and there are three empty consecutive locations.
Thus, the first \explorer, will become \marker $p$ and the other will move towards $p$.
Finally, by Lemma~\ref{lemma:numberexplorer}, it follows that, even if both $p_1$ and $p_2$ are not in $C$, there exists a round in which there are at least three pre.\explorers: two of them will have the same directions; thus, the previous argument can be applied again, and the lemma follows.
\end{proof}

\subriteDescription{Candidate Checking}
\begin{figure}
\begin{framed}
\footnotesize
\begin{algorithmic}[1]	
\Upon[\slave]
	\State $cond_1 := \exists p \in Ports | msg.candidate \in p$
	\State $cond_2 := \nexists p' \in Port | msg.switchtocandidate \in p'$
	\If {$cond_1 \land cond_2$}
	\State $direction=right$
	\State Set State \leader
	\ElsIf  {$cond_1 \land !cond_2$}
	 	\State Set direction and flags from $msg.switchtocandidate$
	\State Set State {\chosen} \label{leader:case2}
		\ElsIf  {$(\exists p \in Ports | msg.switchcollector \in p $}
		 	\State Set direction and flags from $msg.switchcollector$
	\State Set State {\collector} 
			\ElsIf  {$\exists port \in Ports | msg.switchtoleader \in port$} 
		 	\State Set direction and flags from $msg.switchtoleader$
		\State Set State {\chosen} 
	\ElsIf {$\exists port' \in Port | msg.switchtocandidate \in port'$}
	
	\State Set State {\leader}
	\EndIf
\End
\end{algorithmic}
\end{framed}
\caption{Candidate and Slave Algorithm -- Part One \label{algorithm:candidate}}
\end{figure}
\begin{figure}
\begin{framed}
\footnotesize
\begin{algorithmic}[1]
\setcounter{ALG@line}{16}
\Upon[\leader]
	\ForAll{$p \in L_{0} | p \neq \text{\marker}$} \label{candidate:freeze}
	   \State $send(p,msg.switchtoslave)$
	   \EndFor
	   	   \State $cond_1 := \exists p \in getLocation(direction)$
	   \If{$\exists p \in L_0 \land (p=\text{\leader} \lor p=\text{\chosen} ) $}
	   \State Set State {\slave}
           \ElsIf{$cond_1 \land p \neq \text{\leader} \land contracted(p)$} \label{candidate:switch}
	    \State $send(p,msg.switchtocandidate)$
	    \State Set State {\slave}
	    \ElsIf{$contracted(myself) \land cond_1 \land p = \text{\leader}$} \label{leader:sameside1}
	    \State $l=$location not in $L_0$ that is neighbour to $myself$ and $p$.
	    \State $expand(l)$
 	    \ElsIf{$cond_1 \land p = \text{\leader} \land head(p) \notin L_{1,-1}$}   \label{leader:sameside2}
	    \State Set State {\slave}
	    \ElsIf{$head(myself) \notin L_{1,-1} \land !cond_1 \land p = \text{\leader}$}  \label{leader:sameside3}
	    \State $contractToTail()$
	    \State Set State {\chosen}   \label{leader:case5}
	    \Else
	    	\State $cond_2 := \exists p \in L_0 \land p=\text{\marker}$
		\If {$cond_2  \land p.flag.candidate$} \label{cadidate:opposing}
	     	\State $direction=oppositeDirection()$
 	     	\State Set State {\collector} \label{leader:collector}
	     \ElsIf{$cond_2 \land \neg p.flag.candidate \land \nexists msg.asktobecandidate \in p$}\label{leader:ask}
	     \State $send(p,msg.asktobecandidate)$
	     \ElsIf {$\exists p \in Port | msg.candidate \in p$} \label{leader:case1}
	      \State $send(p,msg.switchtoleader)$
	          \State $direction=p$
	       \State Set State {\recruiter}
	     \Else
	     	\State $cond_3 := \exists l_1,l_2,l_3 \in L_0 $ s.t. $empty(l_{1,2,3})=true$ 
		\State $cond_4 := l_1,l_2,l_3$ are contigous
	     	\If{$cond_3 \land cond_4 \land contracted(myself)$} \label{candidate:expand}
		   \State $expand(l_3)$
	    \ElsIf{$cond_3 \land cond_4 \land myhead \in l_3$} \label{leader:case4}
	      \State Set $flag.firstsideswitch$
	   	 \State $direction=getDirectionOppositeToTail()$
	    	\State $contractToHead()$
      	    	\State Set State {\chosen}
	    \ElsIf{$empty(getLocation(direction))$}
	    	\State $move(getLocation(direction))$
	\EndIf
	\EndIf
	  \EndIf
\End
\end{algorithmic}
\end{framed}
\caption{Candidate and Slave Algorithm -- Part Two \label{algorithm:candidateii}}
\end{figure}
The purpose of the {\sf Candidate Checking} sub-algorithm is to determinate whether   one or two \leaderplu have been created.  

A newly  elected \leader has to determine if it is the unique \leader; to do so, it switches 
direction trying to reach the other extremity of the line. While moving, it blocks every neighbour on line $L_{0}$, by sending
a $msg.switchtoslave$ message; this avoids that a new \explorer is created on the portion of the line that has been visited by the \leader (Line~\ref{candidate:freeze} of Figure~\ref{algorithm:candidateii}).

If the \leader  meets an \explorer coming from opposite direction, 
 it ``virtually'' continues its walk by switching roles with the explorer: the \explorer becomes  \leader and switches direction,
and  the old \leader  becomes a  \slave and stops. 
Similarly, if a \leader and a \slave meet, they switch roles (Line~\ref{candidate:switch} of Figure~\ref{algorithm:candidateii}).  

There are two possible outcomes of this procedure: either a unique \chosen is elected, or not.
An unique \chosen can be elected in all the following cases:
\begin{itemize}
 \item (C1) The \leader finds a \marker that has not elected a candidate, flag $flag.candidate$ is unset; also, it sends to \marker a message asking to be a candidate, and it receives the affermative answer from the \marker (Lines~\ref{candidate:expand},\ref{leader:ask},\ref{leader:case1} of Figure~\ref{algorithm:candidateii}). This also implies that  the \leader is unique and no other \leader can be created.
\item (C2) A \slave receives message $msg.candidate$ (from \marker  $p'$),  and message\\ $msg.switchtocandidate$ (see Line~\ref{leader:case2} of Figure~\ref{algorithm:candidate}). 
 \item (C3)  An \explorer receives on two distinct ports a\\ $msg.switchtocandidate$ request. This occurs when there are two \leaderplu on the same side, and both tried to switch their role with the same explorer (Line~\ref{leader:case3} of Figure~\ref{algorithm:explorer}). 
 \item (C4) A \leader reaches the other extremity, it finds three empty locations, and  it expands occupying the last empty location (Line~\ref{leader:case4} of Figure~\ref{algorithm:candidateii}). This implies that the \leader is unique and that no other \leader can be created.
 \item(C5)  A  \leader meets another \leader (Line~\ref{leader:sameside1} of Figure~\ref{algorithm:candidateii}). In this case, each leader knows the position of $L_{0}$, and it can identify the unique location $l$ that is neighbour to both and that is not on $L_0$. $l$ is empty. Both \leader try to expand to $l$; the one that succeed, waits until the other \leader becomes a \slave. When this happens it contracts and it becomes a \chosen (see Lines~\ref{leader:sameside1},\ref{leader:sameside2},\ref{leader:sameside3} of Figure~\ref{algorithm:candidateii}).
\end{itemize}

In all the above cases, the sub-algorithm {\sf Unique Leader} is executed. 
We cannot immediately elect a \chosen   when  the \leader reaches the other extremity,  finding  a \marker that has elected a candidate (Line~\ref{cadidate:opposing} of Figure~\ref{algorithm:candidateii}).
In this last case, the sub-algorithm {\sf Opposite Sides} is executed. It can still happen that, during the execution of {\sf Opposite Sides}, the symmetry between \leader{\em s} is somehow broken: in this case, an unique \chosen is elected, as explained in detail in the section describing {\sf Opposite Sides}.

\begin{lemma} \label{lemma:1candidate}
There exists a round $r$ in which there is at least one \leader.
\end{lemma}

\begin{proof}

By Lemma~\ref{lemma:markerexplorer}, there exists a round when there is at least one \marker $p$, and an \explorer $p'$ moving towards the \marker.
Let this \explorer be in $L_1$.
Notice, that the \explorer can only be stopped by a \leader, turning it into a \slave.  However, if there exists a \leader, the lemma follows. 

Therefore, let us assume that no \leader exists. Notice that, if $p'$ is blocked by another \explorer $p''$ that is pointing at him, then $p''$ receives a message $mgs.changedirection$. 
This message forces $p''$ to eventually get the same direction of $p'$, thus pointing at $p$.  If there is a another
\explorer blocking $p''$, we can iterate the same argument, until there is no one on $L_1$ between the last \explorer $p'$ and a neighbour location of $p$ on $L_1$. Thus, after a finite number of activations $p'$ will be a neighbour of $p$.
When this occurs, either $p'$ becomes a \leader by receiving $msg.candidate$ from $p$; or there exists another \explorer on $L_{-1}$ that asked to become a \leader, it received  $msg.candidate$, thus becoming a \leader. In all cases, the lemma follows.
\end{proof}

\subriteDescription{Unique Leader}
The sub-algorithm {\sf Unique Leader}, reported in Figure~\ref{algorithm:leader}, is executed when an unique leader is elected: it main goal is to let the \chosen   collect every particle and eventually form  a single  line.

Let us assume that the \chosen is on line $L_1$: it moves until it sees a \marker. During its movements, when it finds  some particle $p$ on $L_1$ (either an \explorer or a \slave), it 
forces $p$ to become a \chosen, and it changes its own state to \recruiter. By doing so, a line of moving particles is created: this group moves compact using handovers
starting from the \chosen until the last \recruiter of the line.
In particular, the movements are carried out as follows.

{\bf Virtual Movement.} When a \chosen particle $p$ meets a particle $p'$ that is a \slave (or an \explorer), it forces $p'$ to become \chosen and   it changes its own state to non-\chosen. We will refer to this protocol as a {\it move}, and we will say that the \chosen moves in the position of particle $p'$ (even if there has not been any actual movement of $p'$, but just an exchange of roles). 

{\bf Collection of Particles:} When a particle $p$ starts building a line of moving (and correct) particles, we will say that it is {\sf collecting} the particles. In particular, if $p$ is only interested in collecting particles on its direction of movement, then the collection of a particle $p'$ by $p$ is done by moving virtually to $p'$, while setting the state of $p$ to \recruiter. The effect of this is that $p$ will follow $p'$; also, all other \recruiter (behind $p$) will follow using the handovers. If $p$ wants to collect a particle $p'$ on $L_0$,  then it will send to $p'$ a message $msg.switchtofollower$, and it will wait until $p'$ switches to \recruiter state. After the switch, $p$ keeps moving in its direction; also, $p'$ will join the line by an handover by either $p$ or by some other \recruiter already in the line that is being built by $p$.
 
When a \chosen is a neighbor of a \marker for the first time, the \marker becomes the \chosen, and the \chosen becomes a \recruiter: at this point, the new appointed \chosen is on $L_0$. Now, the \chosen starts moving on $L_{-1}$: during its movements, it makes particles on both $L_0$ (clearly only the correct ones) and on $L_{-1}$ to follow it.
In detail, when the  \chosen, while moving on $L_{-1}$, becomes neighbour of a particle $p$ on $L_{0}$, it sends to this particle a message $msg.switchtofollower$; the \chosen does not move until $p$ becomes \recruiter.
When $p$ becomes \recruiter on $L_0$, it will join the line of the \chosen: in particular, when a tail of an expanded \recruiter or \chosen that is neighbour of $p$ contracts, it does an handover with $p$ forcing it to join the line. 

When the \chosen reaches the other \marker, it switches again side, from $L_{-1}$ to line $L_{1}$, and it keeps ``collecting'' the other correct particles following the same strategy. 
An example run with an Unique Leader is in Figure \ref{xf:shape2}.
\begin{figure*}[tbh]
\begin{center}
  \subfloat[Three \explorers start moving.]{
  \includegraphics[scale=0.45]{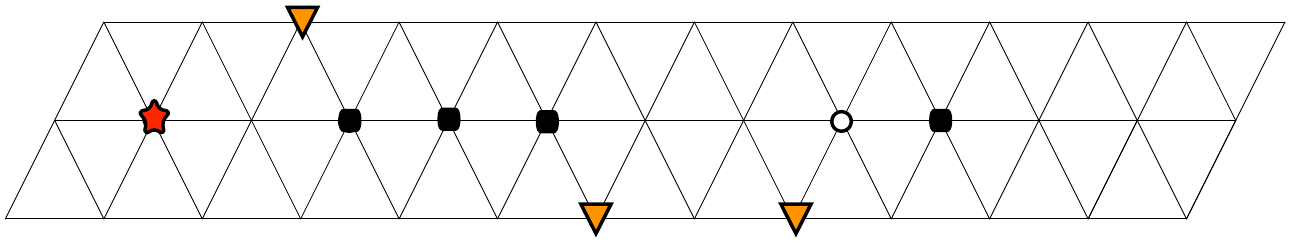}
  }
    \qquad\quad
  \subfloat[All \explorers are moving towards the same \marker, the red star: an unique \chosen will be elected. ]{
  \includegraphics[scale=0.45]{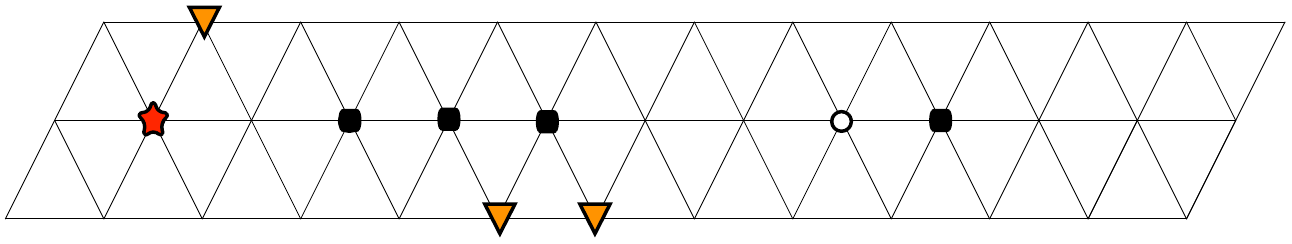}
  }
   \\
  \subfloat[Two explorers ask to become \leader at the same time.]{
  \includegraphics[scale=0.45]{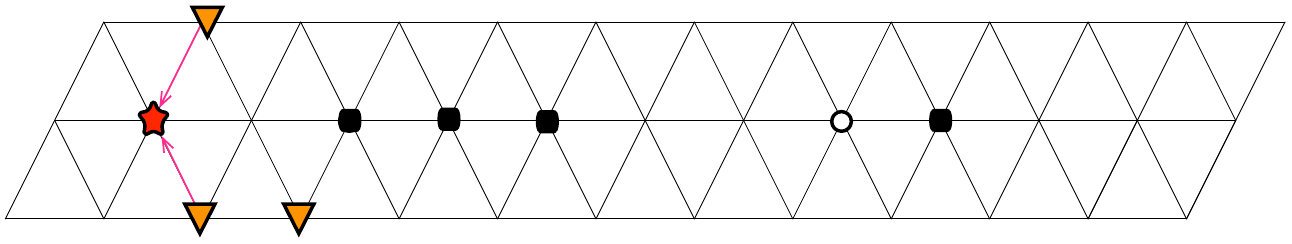}
  }
    \qquad\quad
  \subfloat[A \leader is created, the grey particle with the white border.]{
  \includegraphics[scale=0.45]{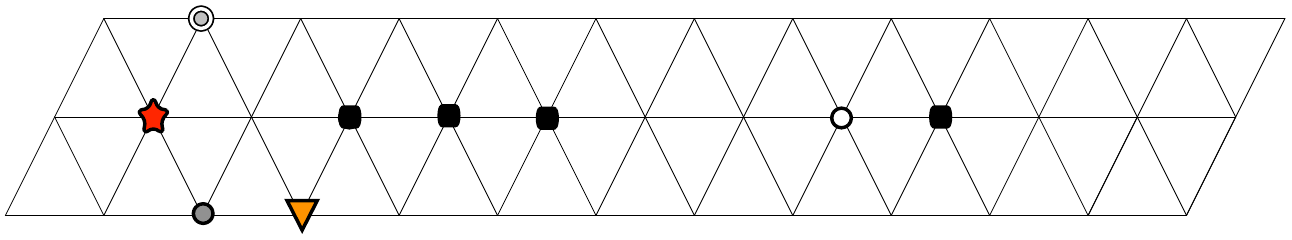}
  }
     \\
  \subfloat[The \leader reaches the end of the line, it detects see by seeing three consecutive locations on $L_0$.]{
  \includegraphics[scale=0.45]{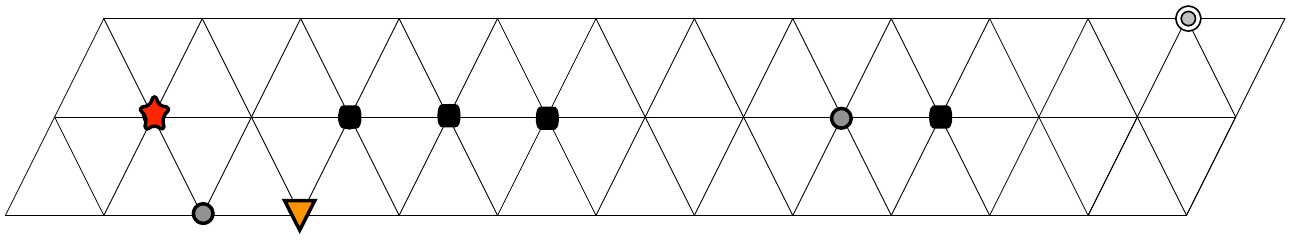}
  }
    \qquad\quad
  \subfloat[The \leader moves on $L_0$ and it becomes an Unique \chosen, this is the case (C4) of the Unique Leader procedure.]{
  \includegraphics[scale=0.45]{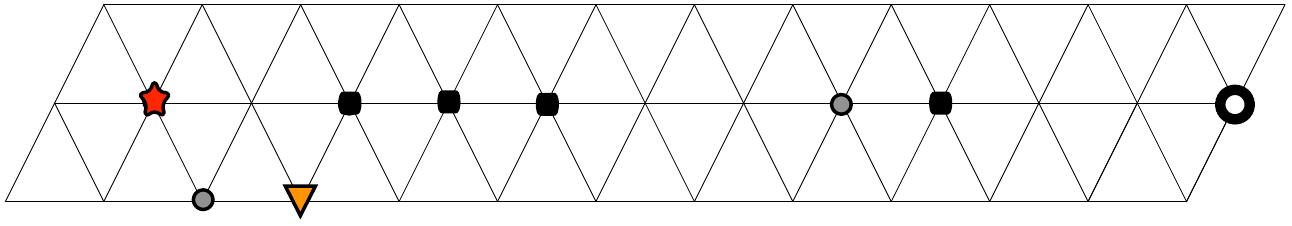}
  }
   \\
  \subfloat[The \chosen moves on $L_{-1}$ collecting the particles it meets, the \recruiter are the grey rectangles.]{
  \includegraphics[scale=0.45]{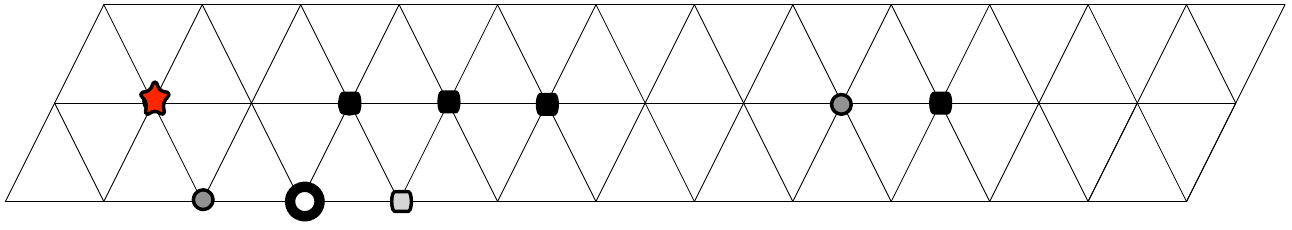}
  }
    \qquad\quad
  \subfloat[The \chosen moves on $L_{1}$ and it starts collecting particles on $L_0,L_1$.  ]{
  \includegraphics[scale=0.45]{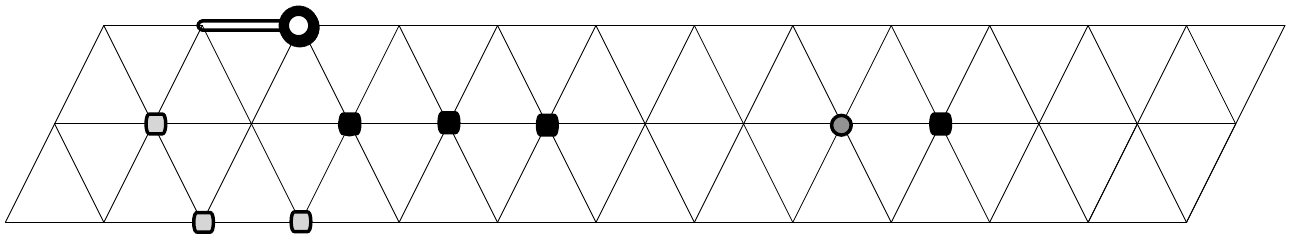}
  }
      \\
  \subfloat[The \chosen starts an handover with a \recruiter on $L_0$.]{
  \includegraphics[scale=0.45]{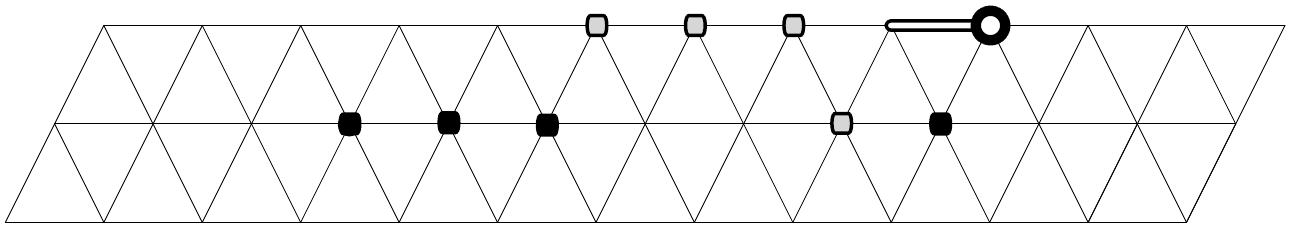}
  }
    \qquad\quad
  \subfloat[All correct particles have been collected by the \chosen.  ]{
  \includegraphics[scale=0.45]{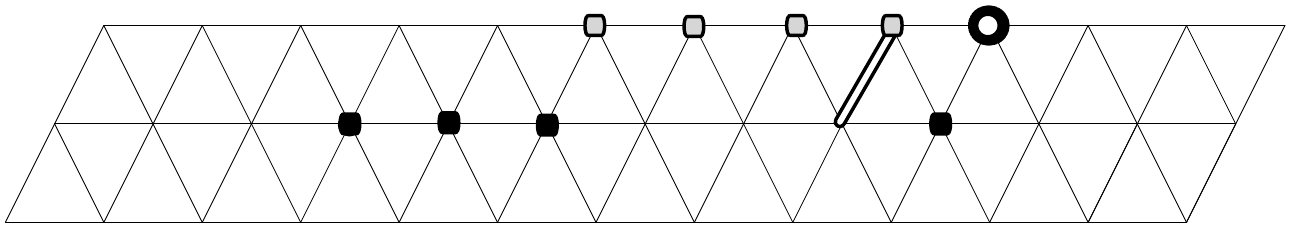}
  }

\end{center}
\caption{Unique Leader, example run.}
\label{xf:shape2}
\end{figure*}

\begin{figure}
\begin{framed}
\footnotesize
\begin{algorithmic}[1]

\Upon[\recruiter]

\If{$expanded(myself)$}
	   \If{$\exists p \in tailNeighbour \land p = \text{\recruiter}$}
	    \State $contractToHeadAndHandover(p)$
	   \Else
	   \State $contractToHead()$ 
	   \EndIf
     	    		\ElsIf  {$\exists port \in Ports | msg.switchtoleader \in port$} 
		 	\State Set direction and flags from $msg.switchtoleader$
		\State Set State {\chosen} 
	\Else
		\State $cond_1 := \exists p \in Port | msg.probe \in p \land contracted(myself)$
		\If {$cond_1 \land contracted(opposite(p))$}
         	\State $send((opposite(p),msg.probe)$
             \ElsIf{$cond_1 \land empty(opposite(p))$} \label{setprobe}
             \State $direction=opposite(p)$
             \State Set State \probe
     \EndIf
     \EndIf
\End

\Upon[\chosen]

           \If{$myself \in L_{j \in \{-1,1\}}  \land getLocation(direction) \notin L_{j}  $ }
           \State $direction=setDirectionToOtherLineExtreme()$
           	    \ElsIf{$contracted(myself) \land flag.firstsideswitch \land \exists p \in L_0 |  p  \neq  \text{\recruiter}$}
	   \State $send(p,msg.switchtofollower)$
           \ElsIf{$\exists p \in getLocation(direction)$} \label{leader:switch}
	    \State $send(p,msg.switchtoleader)$
	    \State Set State {\recruiter}
	    \ElsIf{$\exists p \in L_0 \land p=\text{\marker}$} \label{leader:changeside}
	    \State Set $flag.firstsideswitch$
	    \State $direction=p$
	    \State $send(p,msg.switchtoleader)$
            \State Set State {\recruiter}

            \ElsIf{$expanded(myself)$}
	   \If{$\exists p \in tailNeighbour \land p = \text{\recruiter}$}
	    \State $contractToHeadAndHandover(p)$
	   \Else
	   \State $contractToHead()$ 
	   \EndIf
	    	    \ElsIf{$contracted(myself) \land empty(getLocation(direction))$}
	    \State $expand(getLocation(direction))$
	  \EndIf
\color{black}

\End

\end{algorithmic}
\end{framed}
\caption{Unique Leader\label{algorithm:leader}}
\end{figure}

\begin{lemma}
At any round $r$, there is at most one particle $p$ in state \chosen. \label{lemma:chosen}
\end{lemma}
\begin{proof}
The proof is by case analysis on how a particle becomes \chosen. Let $r$ be the first round at which there is a \chosen.
\begin{itemize}
\item (C1) The \leader particle $p$ on side $L_{1}$ receives a message $msg.candidate$ from a \marker  $p'$ (see Line~\ref{leader:case1} of Figure~\ref{algorithm:candidateii}). By construction, each \marker may send only a $msg.candidate$. If $p$ is a \leader at round $r$, then it received a $msg.candidate$ at round $r' \leq r$. Also, since there are only two \markers, no other \leader  becomes a \chosen by receiving a $msg.candidate$ (either Line~\ref{leader:case1} of Figure \ref{algorithm:candidateii} or Line~\ref{leader:case2} of Figure~\ref{algorithm:candidate}).  Moreover, since $p$ is not anymore a \leader, and since there could be at most two \leader{s}, no particle can become \chosen (by  Line~\ref{leader:case3} of Figure~\ref{algorithm:explorer}), otherwise there would be two \leader{s}. Notice that, after round $r$, the \chosen moves on side $L_{-1}$, blocking any new \explorer. Additionally, no \explorer can be on $L_0$ since the \leader has sent a message $msg.switchtoslave$ to every particle on $L_0$. Therefore, no new \marker can be created after round $r$. This ensures that no new \leader can be created, for any round $r' \geq r$; therefore, there is no  creation of a new \chosen (by Lines~\ref{leader:case1} of Figure \ref{algorithm:candidateii},Line \ref{leader:case2} of Figure~\ref{algorithm:candidate}, and by Line~\ref{leader:case4}  of Figure~\ref{algorithm:candidateii}).

\item (C2) A \slave receives a message $msg.candidate$ from a \marker  $p'$, and a message\\ $msg.switchtocandidate$ (see Line~\ref{leader:case2} of Figure~\ref{algorithm:candidate}). The proof follows an argument similar to the one of the previous case. 
\item (C3) An \explorer particle $p$ on side $L_{1}$ receives two $msg.switchtocandidate$ on two opposing ports (Line~\ref{leader:case3} of Figure~\ref{algorithm:explorer}). Since there could be at  most two \leader{\em s}, and a message  $msg.switchtocandidate$ can only be generated by a \leader, and no other \leader can be created by the \markers after round $r$, it follows that no other particle $p$ becomes \chosen at a round $r' \geq r$ (by  Line~\ref{leader:case3} of Figure~\ref{algorithm:explorer}).
Also, all the \explorer on side $L_{1}$ have been turned into \slave by the two \leaderplu. Moreover, no \explorer can be created after round $r$: every  particle still on $L_{0}$ received the message $msg.switchtoslave$ by one \leader.  The proof follows similarly to previous Case (C1).

\item (C4) A \leader particle $p$, moving on side $L_{1}$ occupies the last of the three empty locations at the end of the line (see Line~\ref{leader:case4}  of Figure~\ref{algorithm:candidateii}). In this case, only one \marker exists, so none of Cases (C1), (C2), and (C3) can be verified at round $\geq r$. Also, the \chosen will move on the other side of the line, blocking any moving \explorer; thus, since all \explorer on $L_1$ and $L_0$ have been blocked by the \leader itself, it is not possible that Case (C4) is verified after  round $r$. 

\item (C5) Two \leaderplu meets (Line~\ref{leader:sameside1} of Figure~\ref{algorithm:candidateii}). It is easy to see that  in this case only one of the \leaderplu is elected. Once the \leader becomes \chosen, the same scenario of previous Case (C4) occurs, the proof follows similarly.
 
\item (C6) A particle $p$ becomes \chosen after receiving the $msg.switchtoleader$: this message can only be sent by a \chosen at a previous round. Therefore, $r$ cannot be the first round at which there is a \chosen.
\end{itemize}

Thus, after round $r$, Cases (C1), (C2), (C3), (C4) and (C5) cannot be verified. The only case left is case (C6): it is easy to see that in this case the number of \chosen{\em s} cannot increase. Therefore, the \chosen is at most one. 
\end{proof}

\begin{theorem}\label{th:chosen}
If there exists  a \chosen at round $r$, then the \algo problem is solved. 
\end{theorem}
\begin{proof}
Let $r$ be the first round at which a \chosen is elected. By Lemma~\ref{lemma:chosen}, the \chosen will be unique. We distinguish the possible cases.

\begin{enumerate}
\item At round $r$, the \chosen particle $p$ is created by Line~\ref{leader:case4} of Figure~\ref{algorithm:candidateii}, and it is moving on side $L_1$. So $p$ occupied an empty location $l \in L_{0}$ at round $r-1$. 
All \explorer{\em s} on side $L_1$ have been blocked by moving on $L_1$ as \leader. Moreover, after round $r-1$ no other \explorer will leave $L_0$ since all particles on line $L_0$ have received  $msg.switchtoslave$. 

Note that, from round $r-1$, any \explorer on side $L_{-1}$ near to location $l$ will switch to state \slave if it sees $p$. Since any \explorer (at distance $1$ from $l$) on side $L_{-1}$ tries to occupy $l$, no \explorer on side $L_{-1}$  can move further location $l$ in a round $r' \leq r-1$ (Line~\ref{candidate:expand} of Figure~\ref{algorithm:candidateii}).
In the following activations, $p$ moves on side $L_{-1}$, towards \marker ($p$ has the flag $flag.firstsideswitch$ set). In particular, the \chosen $p$ collects every non \marker particle on line $L_0$, by sending to each of them a $msg.switchtofollower$ message, and by pulling each one of them on its followers line using an handover; this pulling operation is executed either by the \chosen itself, or by one of the \recruiter{\em s} (there will always be at least one \recruiter on $L_{-1}$ nearby a \recruiter still on $L_0$). Moreover, the \chosen collects also 
each particle on $L_{1}$ (see Line~\ref{leader:switch} of Figure~\ref{algorithm:leader}). 

Note that the \marker has to exists, because otherwise no \leader could have been created. When \chosen reaches the other \marker at the opposite end of the line, it switches side again (Lines~\ref{leader:changeside} of  Figure~\ref{algorithm:leader}). At this point, the \chosen has collected every particle on $L_{-1}$.
From now on, the \chosen will collect any remaining particle on $L_{1}$ and $L_0$; these particles are not moving, since they are all in state either \slave or \recruiter. 
Eventually, the \chosen and its followers will form a straight connected line that includes every particle in $C$.

\item At round $r$, the \chosen particle $p$ at location $l$ on side $L_1$ is generated by Line~\ref{leader:case3} of Figure~\ref{algorithm:explorer}. All particles on $L_1$ are in state \slave, and, since all particles on line $L_0$ have received  {\it msg.switchtoslave}, no other \explorer will leave $L_0$. Similarly to the previous case, the \chosen will reach one of the \marker $p'$, collecting all particles between $l$ and the position of $p'$. 
Meanwhile, no \explorer on side $L_{-1}$ may move further than the \markers (see Line~\ref{explorer:stopcandidate} of Figure \ref{algorithm:explorer}). Now, the \chosen switches side (Lines~\ref{leader:changeside} of  Figure~\ref{algorithm:leader}), and sets $flag.firstsideswitch$. The last part of the the proof is analogous to the previous case. 

\item There are still three possible cases to consider: a \slave  becomes \chosen executing Line~\ref{leader:case2} of Figure~\ref{algorithm:candidate}; a \leader executes Line~\ref{leader:case1} of Figure~\ref{algorithm:candidateii}; and two \leaderplu meets (Line~\ref{leader:case5} of Figure~\ref{algorithm:candidateii}). All of them lead to the same scenario, in which the \chosen $p$ on $L_1$ is nearby a \marker,  all particles on $L_1$ are \slave, and all particle on $L_0$ received $msg.switchtoslave$. The analysis is similar to the previous case.
\end{enumerate}

In all cases, the theorem follows. 
\end{proof}

\subriteDescription{Opposite Sides}
\begin{figure}
\begin{framed}
\footnotesize
\begin{algorithmic}[1]
\Upon[\probe] 
   	    \If{$\exists p \in L_0 \land p=\text{\marker}$} \label{code:probe0}
	     \If{$\exists p \in Port | msg.seen \in p$}
	     \State Set State {\slave}
	     \EndIf
    	    \ElsIf{$contracted(myself) \land empty(getLocation(direction))$}
	    \State $expand(getLocation(direction))$
	    \ElsIf{$expanded(myself)$}
	    \State $contractToHead()$ 
	    \ElsIf{$\exists p \in getLocation(direction) \land contracted(p) \land p=\text{\slave}$}
	    \State $send(p,msg.switchtoprobe)$
	    \State Set State {\slave}
	        	    		\ElsIf  {$\exists port \in Ports | msg.switchtoleader \in port$} 
		 	\State Set direction and flags from $msg.switchtoleader$
		\State Set State {\chosen}  \label{code:probe1}
	  \EndIf
\End
\\
\Upon[\collector]
           \If{$contracted(myself) \land \exists p \in L_0 |  p  \neq  \text{\recruiter}$}
	   \State $send(p,msg.switchtofollower)$
           \ElsIf{$\exists p \in getLocation(direction) \land contracted(p)$}
           \State $send(p,msg.switchtocollector)$
	    \State Set State {\recruiter}
	       \ElsIf{$expanded(myself)$}
	      \If{$\exists p \in tailNeighbour \land p = \text{\recruiter} \land p \in L_0$} \label{collector:pullonlo}
	    \State $contractToHeadAndHandover(p)$
	   \ElsIf{$\exists p \in tailNeighbour \land p = \text{\recruiter}$}
	    \State $contractToHeadAndHandover(p)$
	   \Else
	   \State $contractToHead()$ 
	   \EndIf
	   \Else
	    	\State $checkAngle := angle(location(p),getLocation(direction)) = 60^{\circ}$
	   	\State $cond_1 = \exists p \in L_0 \land p=\text{\marker} \land checkAngle$
		\If{$cond_1 \land \neg flag.firstchange$}  \label{firstchange}
	    		\State Set $flag.firstchange$
	 		\State $direction=oppositeDirection()$
	 	 \ElsIf{$cond_1 \land flag.firstchange$}   \label{coll:counting}
		 	\State Set State {\collector.{\em counting}}
	    	  \ElsIf{$contracted(myself) \land empty(getLocation(direction))$}
	    		\State $expand(getLocation(direction))$
		\EndIf
	  \EndIf
\End
\end{algorithmic}
\end{framed}
\caption{Collector Algorithm -- Part One\label{algorithm:collector}}
\end{figure}
\begin{figure}
\begin{framed}
\footnotesize
\begin{algorithmic}[1]
\setcounter{ALG@line}{40}
\Upon[\collector{\em .counting}]
   \If{$\neg flag.firstrun$}
    \State $send(myself,msg.other)$ \label{bootstartp:count}
   \State Set $flag.firstrun$
   \ElsIf{$\exists p \in Port | msg.other \in p$} \label{newprobe}
   	\If{$\exists p' \in Neighbours \land p'=\text{\recruiter}$}
	\State $send(p',msg.probe)$
	\Else
	\State $\nexists p'' | p'' \in Neighbours \land p''= \text{\recruiter}$ \label{nomoreprobe}
	\State $send(p'',msg.done)$
	\State $direction=oppositeDirection()$
	\State Set State \text{\collector{\em .done}}
	\EndIf
	   \ElsIf{$\exists p \in Port | msg.winner \in p$}
	   	\State Set State \text{\chosen}
   \EndIf
\End
\\
\Upon[\collector{\em .done}] 
 	 \If{$\exists p \in L_0 \land (p=\text{\marker} \lor p= \text{\chosen})$} \label{collecotrdone1}
	 	    \State End of Cycle
	         \ElsIf{$\exists p \in Port | msg.even \in p$}
	     \State Set State {\chosen}
    	    \ElsIf{$contracted(myself) \land empty(getLocation(direction))$}
	    \State $expand(getLocation(direction))$
	    \ElsIf{$expanded(myself)$}
	    \State $contractToHead()$ 
	    \ElsIf{$\exists p \in getLocation(direction) \land contracted(p) \land p=\text{\slave}$}
	    \State $send(p,msg.switchtocollectordone)$
	    \State Set State {\slave}
	        	    		\ElsIf  {$\exists port \in Ports | msg.switchtoleader \in port$} 
		 	\State Set direction and flags from $msg.switchtoleader$
		\State Set State {\chosen}   \label{collecotrdone2}
	  \EndIf
\End
\end{algorithmic}
\end{framed}
\caption{Collector Algorithm -- Part Two\label{algorithm:collectorii}}
\end{figure}
This sub-algorithm starts when a \leader on $L_{1}$ (respectively,  $L_{-1}$) reaches a \marker $p$ with flag $flag.candidate$ set: the \leader realizes that there is another \leader  moving in the same direction (either clockwise or counter-clockwise) on $L_{-1}$ (respectively, $L_{1}$) (see Lines \ref{cadidate:opposing}-\ref{leader:collector} of Figure \ref{algorithm:candidateii}). Also, the \leader becomes \collector, switches direction, and moves towards the other \marker  $p'$; during this movements, the \collector forces every particle it encounters (also on $L_0$) to become a \recruiter. Moreover, when contracting, it pulls any \recruiter near to its tail.  

Once the \collector reaches $p'$, it reverts again direction to reach $p$ (Line \ref{firstchange} of Figure \ref{algorithm:collector}); once it reaches $p$, it switches state to \collector{\em .counting} (Line \ref{coll:counting} of Figure \ref{algorithm:collector}).
Let us assume that the \collector has recruited at least one particle. 
The \collector{\em.counting} propagates a message $msg.probe$ on its  \recruiter{\em s} on $L_1$ (Line \ref{bootstartp:count} of Figure \ref{algorithm:collectorii}). Once the $msg.probe$ reaches a \follower $p''$ with only one neighbour, particle $p''$ switches state to \probe (Line \ref{setprobe} of Figure \ref{algorithm:leader}). The \probe travels to reach \marker $p'$, and it temporarily stops  if it sees another \probe. The code for the \probe is in Figure \ref{algorithm:collector} Lines \ref{code:probe0}-\ref{code:probe1}. When \marker $p'$ sees a \probe and a \collector{\em .counting}, it sends to the \probe a message $msg.seen$. The \probe switches state to \slave when it reads such message. 

For each $msg.seen$ sent, the \marker sends a message $msg.other$ to a neighbour \collector{\em.counting}. 
When the \collector{\em.counting} receives a $msg.other$ from a \marker, it propagates a new message $msg.probe$ (Line \ref{newprobe} of Figure \ref{algorithm:collectorii}).
 Eventually, the \collector{\em.counting} receives $msg.other$ and it sees that it does not have any \recruiter (Line \ref{nomoreprobe} of Figure \ref{algorithm:collectorii}); in this case, it switches state to \collector{\em.done} and it moves until it becomes neighbour of $p'$. Before moving, it notifies the \marker, by sending a message $msg.done$.
While the \collector{\em.done}  moves, it turns to \recruiter any \slave it meets on its way, and it waits if it sees a \probe (The code fo the \collector{\em.done} is in Lines \ref{collecotrdone1}-\ref{collecotrdone2} of Figure \ref{algorithm:collectorii}).
If the \marker sees a \collector{\em .done} and a \collector{\em .counting}, that does not have a $msg.other$ pending, it sends $msg.winner$ to it (Line \ref{collectorwinner} of Figure \ref{algorithm:activated}). Otherwise, if the \marker sees only a \collector{\em .done} and it was notified by a $msg.done$, then it sends $msg.even$ to the \collector{\em .done} and it becomes a \recruiter (Line \ref{collectoreven} of Figure \ref{algorithm:activated}). 
If a \collector{\em .done} receives a $msg.even$, then it switches state to \chosen.
If a \collector{\em.counting} receives a $msg.winner$, then it switches state to \chosen.
An example run with \collector{\em s} is in Figure \ref{xf:shape3}.

\begin{figure*}[ht]
\begin{center}
  \subfloat[Three \explorers start moving.]{
  \includegraphics[scale=0.5]{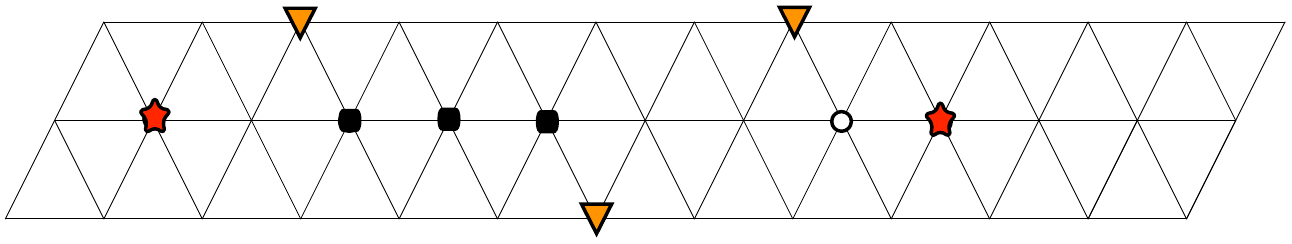}
  }
    \qquad\quad
  \subfloat[ \explorers are moving towards different \marker{\em s}: we will one \leader for each side. ]{
  \includegraphics[scale=0.5]{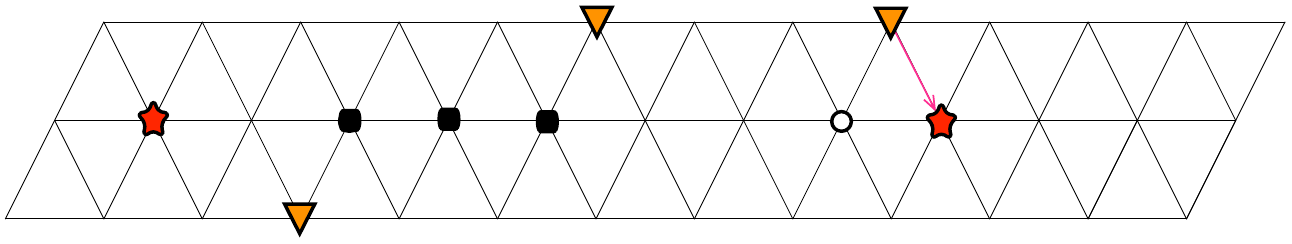}
  }
   \\
  \subfloat[ A \leader is created on $L_{1}$ and an \explorer asks to \marker to became a \leader on $L_{-1}$.]{
  \includegraphics[scale=0.5]{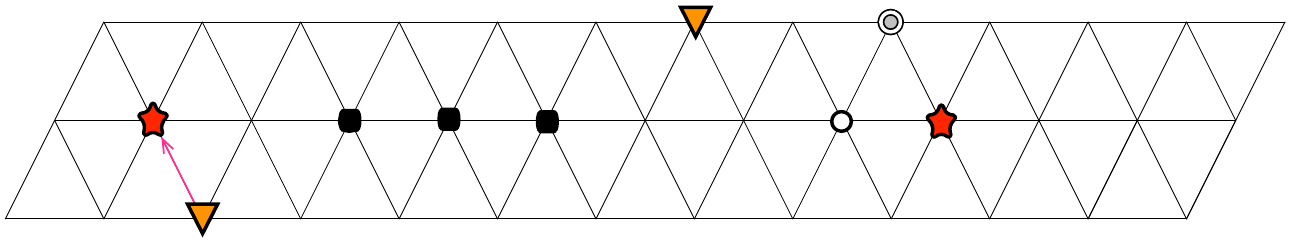}
  }
    \qquad\quad
  \subfloat[Each \leader checks for the presence of the other, they will both become a \collector.]{
  \includegraphics[scale=0.5]{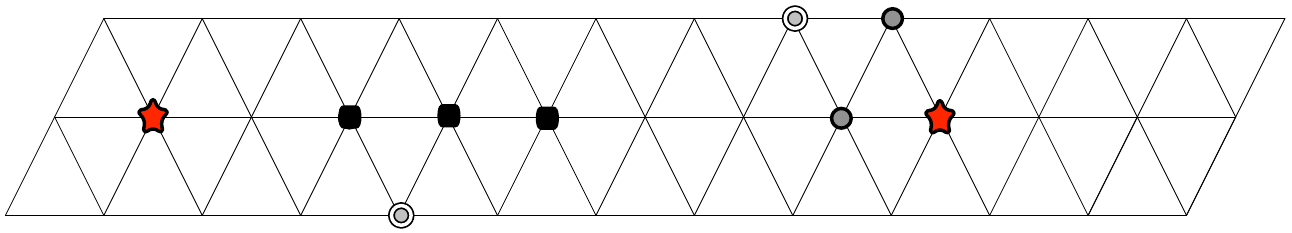}
  }
  \\
  \subfloat[The \collector have collected any correct particle.]{
  \includegraphics[scale=0.5]{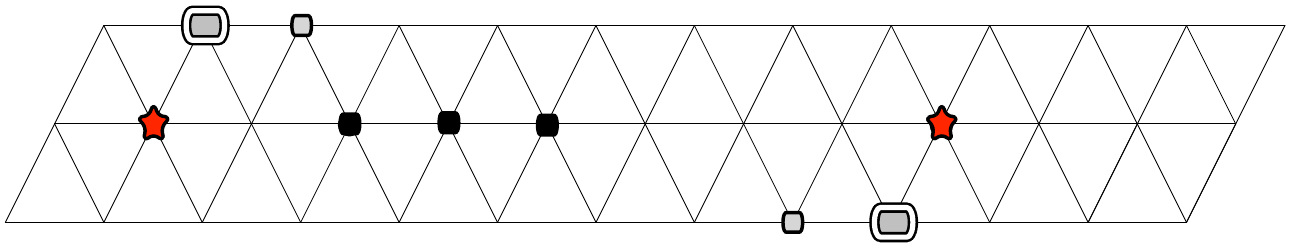}
  }
  \qquad\quad
  \subfloat[The \collector.{\em counting}'s send probes to check if the particles have been divided equally.]{
  \includegraphics[scale=0.5]{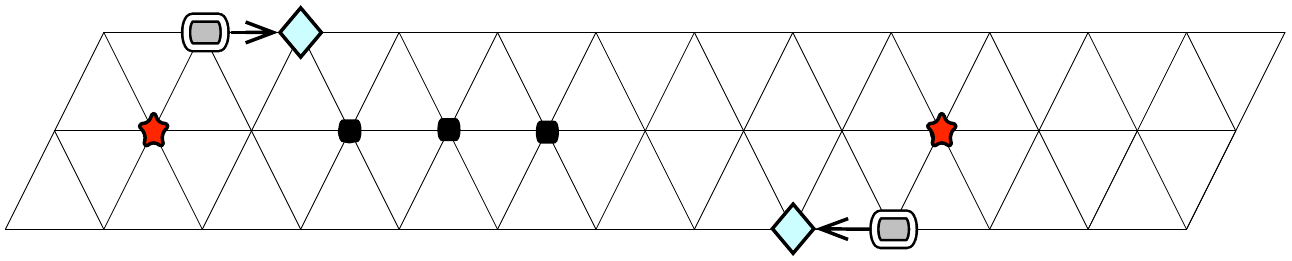}
  }
  \\
  \subfloat[Each \collector.{\em counting} sees that it has no other probe.  ]{
  \includegraphics[scale=0.5]{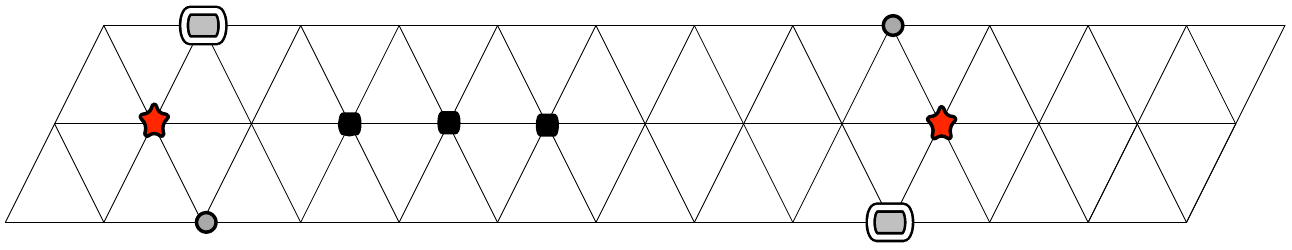}
  }
      \qquad\quad
  \subfloat[Each \collector.{\em counting} becomes \collector.{\em done} and it starts moving towards the other \marker.]{
  \includegraphics[scale=0.5]{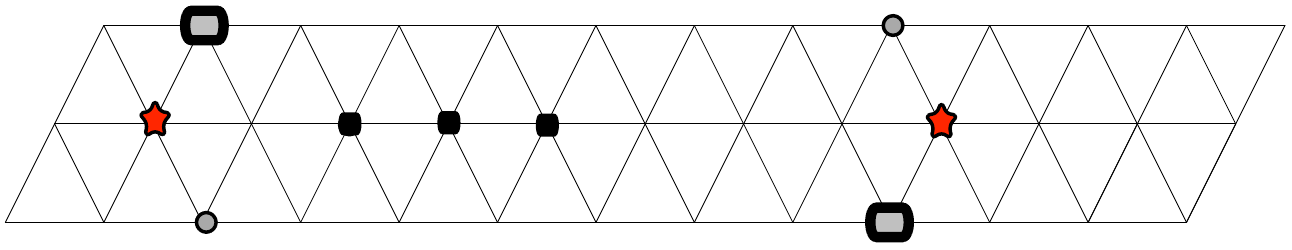}
  }
  \\
  \subfloat[Each \collector.{\em done} reaches a \marker and it is notified by the \marker that also the other \collector was in state \collector.{\em done}. This implies that particles are equally divided in two sets.]{
  \includegraphics[scale=0.5]{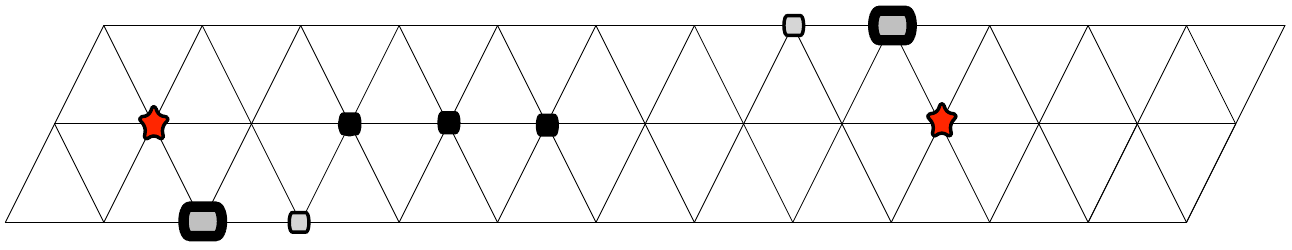}
  }
  \qquad\quad  
  \subfloat[The \marker becomes a \recruiter and each  \collector.{\em done} a \chosen with the moving direction indicated by the arrow. Thus two equa lines will be formed]{
  \includegraphics[scale=0.5]{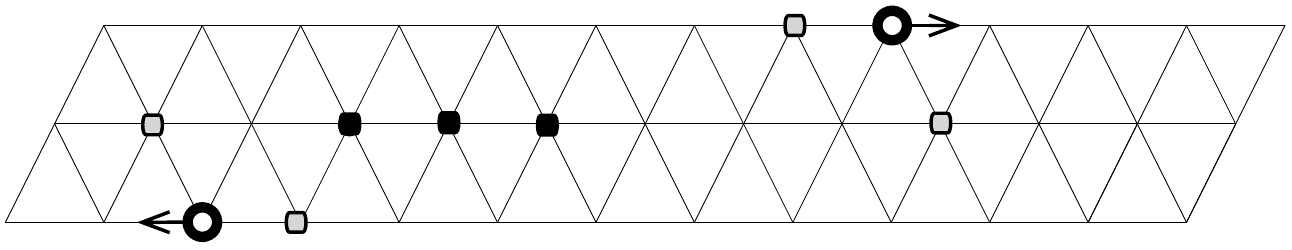}
  }
\end{center}
\caption{Opposite Sides, example run.}
\label{xf:shape3}
\end{figure*}

\begin{theorem}\label{th:collector}
If there exists  a \collector at round $r$, then the \algo problem is solved. 
\end{theorem}
\begin{proof}
If a \leader, on side $L_1$, becomes a \collector when it reaches \marker $p$, then there exists a \leader on side $L_{-1}$  that will become \collector when it reaches \marker $p'$; let $c'$ be such \collector.
The \collector $c$ moves from $p$ to $p'$, and then back to $p$. Let $C'$ be the set of particle in $C$ excluding the \marker{\em s} and the \collector{\em s}.
When $c$ reaches $p$ for the second time, it switches state to \collector{\em .counting}; similarly, also $c'$ will eventually reach $p'$, switching state to \collector{\em .counting}. In the following, we will show that, when this occurs, all particles in $C'$ are followers of either $c$ or $c'$, thus forming two segments $s$ and $s'$, respectively. Note that, for each of the  \collector{\em .counting} followers, it might be that one of the particles could be a \probe moving to a \marker. For now, let us assume that this \probe is still belonging to its initial segment.  

Also, let us assume that there exists a particle $p^*$ that is still in $L_0$ and that does not belong to $c$ or $c'$; that is, no one pulled $p^*$. However, when $c$ is going back from $p$ to $p'$, no particle is on $L_1$; therefore, there will be a round in which $c$ is expanded and its tail is neighbour of $p^*$. By construction, when $p$ contracts, it will pull $p^*$ (Line~\ref{collector:pullonlo} of Figure~\ref{algorithm:collector}).  For particles on $L_1$ and $L_{-1}$ the theorem easily follows, since they will be collected when the corresponding \collector is going from the \marker on which it switched state to the other one. 

In the the next step of the sub-algorithm, the sizes of $s$ and $s'$ are compared by using the \probes. Each segment decreases by $1$ by moving a \probe to the next \marker, and then it waits for a signal from the probe of the other \collector to further decrease. We have two possible cases:

\begin{enumerate}
\item If $|s| \leq |s'|$, it is easy to see that $c$ first receives a $msg.other$, and then it becomes a \collector{\em .done} moving towards $p'$.
Thus the \marker $p'$ sees a \collector{\em .counting} $c'$ and a \collector{\em .done} $c$; therefore, $c'$ becomes a \chosen. From now on, it is easy to see that $c'$ will collect every particle on its way: this case is similar to the case of the {\sf Unique Leader} sub-algorithm (refer to Therem~\ref{th:chosen}).

\item If $|s| = |s'|$, both \collector{\em .counting} will receive a $msg.other$ while there is no \follower in the neighbourhood. 
Since the scheduler is semi-synchronous, it might be that, before $c'$ it is activated,  $c$ becomes  \collector{\em .done}  and it reaches the \marker.
In this case, if the \marker sends the $msg.winner$ to $c'$, upon activation $c'$ will find this message and it will become \chosen, and previous case applies.
Otherwise, $c'$ sent a $msg.done$ to the \marker; when $c$ reaches $p'$, the \marker, upon activation, sends $msg.even$ to $c$.

Since $|s|=|s'|$ and while $c$ is moving it turns every \slave into a \recruiter, we have that half of the processes in $C'$ are followers of $c$ and connected. Also, when the \marker sends $msg.even$ to $c$ it also becomes a \follower.
At this point, $c$ becomes a \chosen and it moves, eventually bringing $p'$ in a straight line. The same will be done by $c'$. Therefore, in this case we will have two lines of equal size, and the theorem follows.
\end{enumerate}\end{proof}

Finally, we have:

\begin{theorem}
Starting from any initial configuration, the \algo problem is solved.
\end{theorem}
\begin{proof}
If $f=0$, the theorem follows by Theorem~\ref{thf:0}. Let us thus consider the case where $f >0$. By Lemma~\ref{lemma:1candidate}, a leader will eventually be elected. 

By construction, the \leader creation always occurs at a marker $p$, and then it moves towards the other extreme of the line. Let us suppose that an \explorer sees the extremity of the initial line, and let us consider the first round $r'$ at which this occurs. Since before round $r'$ no \chosen or \collector can be created, by Lemma~\ref{lemma:1candidate} three consecutive empty locations correctly mark the end of the line. Therefore, a \leader can correctly recognise the end of the line either by the presence of a \marker, or by three consecutive empty locations. 

Three cases can occur:

\begin{enumerate} 
\item[(1)] There is an empty location that is occupied by \leader: in this case, {\sf Unique Leader} sub-algorithm is run, a \chosen is elected, and, by Theorem~\ref{th:chosen}, the theorem follows.
\item[(2)] There is \marker that has not elected a \leader, and the leader locks the \marker by receiving a $msg.switchtocandidate$: in this case,  {\sf Unique Leader} sub-algorithm is run, a \chosen is elected, and, by Theorem~\ref{th:chosen}, the theorem follows.
\item[(3)] There is \marker that has elected a \leader: in this case, {\sf Opposite Sides} is run, and, by Theorem~\ref{th:collector}, the theorem follows.
\end{enumerate}
Notice that, it might also be possible that a leader reaches the end of the line marked by an \empty location and that another particle occupies it because of the semi-synchronous scheduler. In this case, a \marker is created, and either Case (2) or (3) above applies. Therefore, in all possible cases, the theorem follows. 
\end{proof}

\bibliographystyle{plain}

  

\end{document}